\documentclass[onecolumn]{IEEEtran}

\usepackage{graphicx}
\usepackage{amsmath,amssymb}
\usepackage{mathtools}
\usepackage{amsthm}
\usepackage{dsfont}
\usepackage{physics}
\usepackage[colorlinks=true,citecolor=blue,linkcolor=blue,urlcolor=blue,]{hyperref}

\newtheorem{definition}{Definition}
\newtheorem{theorem}[definition]{Theorem}
\newtheorem{proposition}[definition]{Proposition}
\newtheorem{corollary}[definition]{Corollary}
\newtheorem{lemma}[definition]{Lemma}

\newcommand{\1}{\mathds{1}}
\newcommand{\dbra}[1]{\langle\!\langle {#1} \vert}
\newcommand{\dket}[1]{\vert {#1} \rangle\!\rangle}
\NewDocumentCommand{\dbraket}{m g}{%
  \IfNoValueTF{#2}
    {\langle\!\langle {#1} \vert {#1} \rangle\!\rangle}
    {\langle\!\langle {#1} \vert {#2} \rangle\!\rangle}%
}
\NewDocumentCommand{\dketbra}{m g}{%
  \IfNoValueTF{#2}
    {\vert #1 \rangle\!\rangle\!\langle\!\langle #1 \vert}
    {\vert #1 \rangle\!\rangle\!\langle\!\langle #2 \vert}%
}
\newcommand{\U}{\mathrm{U}}
\newcommand{\Path}[1]{\mathrm{Path}(#1)}
\newcommand{\isometry}[2]{W_{#2, #1}}
\newcommand{\Comm}{\mathrm{Comm}}

\newcommand{\ZZ}{\mathbb{Z}}
\newcommand{\YY}{\mathbb{Y}}

\newcommand{\mcL}{\mathcal{L}}
\newcommand{\mcX}{\mathcal{X}}
\newcommand{\mcY}{\mathcal{Y}}
\newcommand{\mcZ}{\mathcal{Z}}
\newcommand{\mcU}{\mathcal{U}}
\newcommand{\mcS}{\mathcal{S}}
\newcommand{\mcI}{\mathcal{I}}
\newcommand{\mcO}{\mathcal{O}}
\newcommand{\mcP}{\mathcal{P}}
\newcommand{\mcF}{\mathcal{F}}
\newcommand{\mcH}{\mathcal{H}}
\newcommand{\mcA}{\mathcal{A}}

\newcommand{\mcV}{\mathcal{V}}
\newcommand{\mcD}{\mathcal{D}}
\newcommand{\mcW}{\mathcal{W}}

\newcommand{\mcE}{\mathcal{E}}
\newcommand{\CC}{\mathbb{C}}
\newcommand{\mfS}{\mathfrak{S}}
\newcommand{\sfS}{\mathsf{S}}

\newcommand{\mD}[1]{m_{#1}^{(D)}}
\newcommand{\md}[1]{m_{#1}^{(d)}}

\begin{document}

\title{Optimal complex conjugation of unknown isometry channels}

\author{Satoshi~Yoshida and Mio~Murao
\thanks{S. Yoshida and M. Murao are with the Department of Physics, Graduate School of Science, The University of Tokyo, Tokyo 113-0033, Japan (e-mail: satoshiyoshida.phys@gmail.com; murao@phys.s.u-tokyo.ac.jp).}
\thanks{M. Murao is also with the Trans-scale Quantum Science Institute, The University of Tokyo, Bunkyo-ku, Tokyo 113-0033, Japan.}}

\maketitle

\begin{abstract}
    Access to the complex conjugate of an unknown quantum channel is a useful resource in quantum oracle problems, motivating the question of how such access can be simulated using only a limited number of calls to the original channel.
    We determine the optimal deterministic protocol for approximately implementing the complex conjugate isometry $\overline{V}$ from $n$ uses of an unknown isometry channel $V: \CC^d\to\CC^D$.
    We derive a closed-form expression for the optimal fidelity and prove that a parallel protocol is optimal even among general quantum superchannels, including adaptive and indefinite-causal-order strategies.
    The formula implies a query complexity $n=\Theta(d[(D-d)/\epsilon+1])$ for achieving infidelity $\epsilon$.
    We also present a circuit construction based on the quantum Schur transform and the dual Clebsch--Gordan transform, with circuit complexity $O(\mathrm{poly}(D,1/\epsilon))$.
    This task is extended to the multi-copy case $V^{\otimes n}\mapsto \overline{V}^{\otimes k}$.
    For fixed $d<D$ and $k$, we show that the optimal fidelity for the multi-copy case is $1-kd(D-d)/n+o(n^{-1})$, and that this value is asymptotically attained by a parallel estimation-based protocol.
    Finally, combining the isometry protocol with random Stinespring dilations yields a protocol for complex conjugation of unknown rank-$r$ quantum channels whose query complexity is optimal up to a constant factor if the Kraus rank $r$ is constant.
\end{abstract}

\begin{IEEEkeywords}
Isometry channels, Higher-order quantum transformations, Quantum superchannels, Schur--Weyl duality.
\end{IEEEkeywords}

\section{Introduction}
\label{sec:introduction}

Given only black-box uses of an unknown quantum channel, which related operations can be implemented, and at what query cost?
This question underlies the comparison of oracle access models: an oracle for a unitary channel, its inverse, a controlled one, or its complex conjugate can provide genuinely different computational resources~\cite{cotler2021revisiting,huang2022quantum,vanapeldoorn2023quantum,tang2025conjugate,tang2025controlled,tang2025amplitude}.
Complex conjugation is a particularly natural test case, because conjugate access can substantially reduce the query complexity of oracle problems such as reality testing~\cite{cotler2021revisiting,huang2022quantum,tang2025conjugate}.
It is therefore important to determine when a conjugate query can itself be implemented from forward queries alone.
This problem is considered within the framework of universal transformation of unknown quantum states or quantum channels, which transforms an unknown input into a desired output without a classical description of the input.

Complex conjugation of unknown pure states~\cite {buzek1999optimal,brzic2026optimal} and unitary channels~\cite {miyazaki2019complex,ebler2023optimal} are well understood, and their optimal performance is known.
Isometry channels interpolate between pure-state preparation and unitary channels and arise naturally in various quantum-information settings, including Stinespring dilations of quantum channels and quantum information encoding.
Recently, the channel access model was shown to be equivalent to the Stinespring dilation isometry access model~\cite{girardi2025random,yoshida2025random}.
Still, the lifting to Stinespring dilation unitary requires a dimension-dependent overhead in the query complexity~\cite{yoshida2025random}.
Despite the importance of isometry channels, the optimal performance of isometry complex conjugation has not been known.
More broadly, such a transformation is considered in the framework of quantum superchannels\footnote{They are also called quantum supermaps~\cite{chiribella2008transforming} and higher-order quantum transformations~\cite{taranto2025higher}. See also related formalisms: quantum combs~\cite{chiribella2008quantumcircuits}, quantum channels with memory~\cite{kretschmann2005quantum}, quantum strategy~\cite{gutoski2007toward}, process matrices~\cite{oreshkov2012quantum}, process tensors~\cite{taranto2019quantum}, or quantum algorithmic measurements~\cite{aharonov2022quantum}.}, which is a general framework for transforming unknown quantum channels, e.g., unitary inversion $U\mapsto U^\dagger$, unitary transposition $U\mapsto U^\top$~\cite{chiribella2016optimal,sardharwalla2016universal,hahn1950spin,freeman1998spin,sedlak2019optimal,quintino2019probabilistic,quintino2019reversing,miyazaki2019complex,ebler2023optimal,ishizaka2008asymptotic,ishizaka2009quantum,studzinski2017port,quintino2022deterministic,trillo2020translating,trillo2023universal,yoshida2023reversing,chen2024quantum}, and unitary complex conjugation $U\mapsto \overline{U}$~\cite{miyazaki2019complex,ebler2023optimal} for $U\in\U(d)$, and lower bounds on the query complexity of these tasks are investigated in Refs.~\cite{odake2025analytical,chen2026approximation}.
Unitary inversion is extended to isometry inversion $V\mapsto V^{-1}$~\cite{yoshida2023universal} or isometry adjointation $V\mapsto V^\dagger$~\cite{yoshida2025universal}, and unitary transposition is extended to isometry transposition $V\mapsto V^\top$~\cite{yoshida2023universal} for isometries $V:\CC^d\to\CC^D$, but the extension of unitary complex conjugation to isometry complex conjugation $V\mapsto \overline{V}$ remained open.
This problem is not a straightforward extension of unitary complex conjugation.
The unitary complex conjugation is made possible by using the representation-theoretic analysis of the complex conjugate representation $\overline{U}$~\cite{miyazaki2019complex}.
At the same time, the set of isometry operators does not form a group, so the same representation-theoretic argument is not directly applicable.
Moreover, unlike unitary complex conjugation, the complex conjugation of an unknown pure state cannot be implemented exactly using any finite number of queries, even probabilistically~\cite{yang2014certifying}.
Since pure-state preparation is a special case of an isometry channel, exact complex conjugation of an unknown isometry channel is likewise impossible with finitely many queries~\cite{yoshida2023universal}.
We therefore investigate the optimal approximate implementation of complex conjugation for unknown isometry channels.

This work investigates the optimal deterministic protocol for complex conjugation of an unknown isometry $V:\CC^d\to\CC^D$, which is to implement $\overline{V}$ from $n$ calls of $V$ (see Fig.~\ref{fig:task}), where $\overline{V}$ is the complex conjugate of $V$ in the computational basis.
Our main result provides a closed formula for the optimal fidelity of the task, showing that the parallel protocol attains the optimal fidelity among all possible protocols, including adaptive or indefinite-causal-order ones~\cite{hardy2007towards,oreshkov2012quantum,chiribella2013quantum}.
We provide the optimal fidelity $F$ by using the semidefinite programming and the Schur--Weyl duality, which attains the fidelity $F= 1-\epsilon$ with $n = \Theta(d[(D-d)/\epsilon+1])$ queries, and construct the optimal parallel protocol by using the quantum Schur transform and the dual Clebsch--Gordan transform~\cite{harrow2005applications,bacon2006efficient,bacon2007quantum,krovi2019efficient,burchardt2025high,nguyen2023mixed,grinko2023gelfand,fei2023efficient,grinko2025mixed}, with the circuit complexity $O(\mathrm{poly}(D,1/\epsilon))$.
This circuit construction interpolates between the optimal circuit for transposition of an unknown pure state~\cite{brzic2026optimal} and that for unitary complex conjugation~\cite{miyazaki2019complex,ebler2023optimal}, where the former is given similarly but the circuit complexity of the latter was not known.
The task is extended to the multi-copy case $V^{\otimes n}\mapsto \overline{V}^{\otimes k}$, and we show that the asymptotic optimal fidelity for fixed $d,D,k$ is given by $F = 1-kd(D-d)/n+o(n^{-1})$, which is achieved by a parallel estimation-based protocol using the isometry estimation protocol in Ref.~\cite{yoshida2026quantum}.
By combining the optimal isometry complex conjugation protocol with a random dilation superchannel (corresponding to $\Lambda\mapsto \overline{V}_\Lambda\mapsto \overline{\Lambda}$ in Fig.~\ref{fig:access-models}), we also provide a protocol for complex conjugation of an unknown channel $\Lambda:\mcL(\CC^{d_\mathrm{in}})\to\mcL(\CC^{d_\mathrm{out}})$ of bounded Kraus rank $r$ with $O(d_\mathrm{in}(\Delta_\Lambda/\epsilon+1))$ queries with $\Delta_\Lambda\coloneqq rd_\mathrm{out}-d_\mathrm{in}$, where $\overline{\Lambda}(\cdot)\coloneqq \overline{\Lambda(\overline{\cdot})}$ (see Fig.~\ref{fig:access-models}).
This construction has an advantage over the protocol based on the unitary complex conjugation and the lifting from the channel to its Stinespring dilation unitary (corresponding to $\Lambda\mapsto V_\Lambda\mapsto U_\Lambda\mapsto \overline{U}_\Lambda \mapsto \overline{\Lambda}$ in Fig.~\ref{fig:access-models}), which requires $O(d_\mathrm{in}^2d_\mathrm{out}^2r^2/\epsilon)$ queries to $\Lambda$.
We also show a lower bound on the query complexity of the channel complex conjugation given by $\Omega(d_\mathrm{out}(\min\{d_\mathrm{in}, d_\mathrm{out}, \Delta_\Lambda\}/\epsilon+1))$ for $0<\epsilon\leq 1/4$.
If the Kraus rank is constant ($r=O(1)$), this construction provides the optimal query complexity of the task given by $\Theta(d_\mathrm{in}(\Delta_\Lambda/\epsilon+1))$.

\begin{figure}
    \centering
    \includegraphics[width=0.5\linewidth]{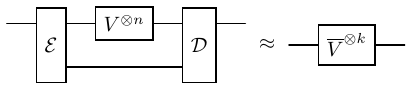}
    \caption{Illustration of the task of isometry complex conjugation.
    By using $n$ parallel queries of an unknown isometry channel $\mcV(\rho)=V\rho V^\dagger$ with $V:\CC^d\to\CC^D$, we implement a channel approximating $k$ parallel queries of the complex conjugate channel $\overline{\mcV}(\rho)=\overline{V}\rho V^\top$.
    We illustrate a parallel protocol with an encoder $\mathcal{E}$ and a decoder $\mathcal{D}$.
    The optimal protocol is shown to be parallel even within the general class of deterministic supermaps for the case of $k=1$, which also contains adaptive and indefinite-causal-order strategies~\cite{hardy2007towards,oreshkov2012quantum,chiribella2013quantum}.
    For $k\geq 1$, the asymptotic optimality of the estimation-based parallel protocol is shown for fixed $d,D,k$ and $n\to\infty$.}
    \label{fig:task}
\end{figure}

\begin{figure}
    \centering
    \includegraphics[width=0.5\linewidth]{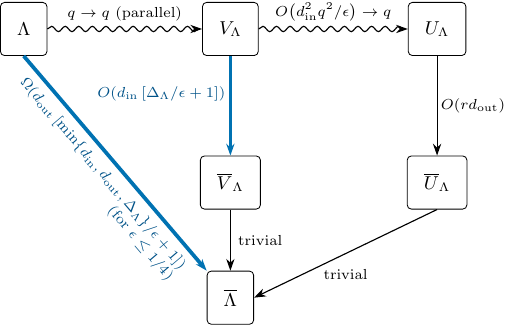}
    \caption{
        Application of this work to complex conjugation of unknown quantum channels.
        Each arrow indicates a lifting from one access model to another, which can be implemented with a certain overhead in the query complexity, and blue arrows are constructed in this work.
        The wave arrow indicates a randomized lifting, which implements $q$ queries of the randomly chosen target access (called the acorn trick~\cite{tang2025conjugate}).
        The access to rank-$r$ unknown quantum channel $\Lambda: \mcL(\CC^{d_\mathrm{in}}) \to \mcL(\CC^{d_\mathrm{out}})$ is lifted to the access to its Stinespring dilation isometry $V_\Lambda$ such that $\Tr_\mcE[V_\Lambda(\cdot)V_\Lambda^\dagger] = \Lambda(\cdot)$ using an environmental system $\mcE\simeq \CC^r$.
        For $q\in \ZZ_{+}$, $q$ parallel queries to a randomly chosen dilation isometry $V_\Lambda$ can be implemented using $q$ parallel queries to the same unknown channel $\Lambda$~\cite{girardi2025random,yoshida2025random}.
        It is further lifted to the access to its Stinespring dilation unitary $U_\Lambda$ such that $\Tr_\mcE[U_\Lambda(\cdot\otimes \ketbra{0})U_\Lambda^\dagger] = \Lambda(\cdot)$, but this lifting requires $O(d_\mathrm{in}^2 q^2/\epsilon)$ queries to $V_\Lambda$~\cite{yoshida2025random} to implement $q$ parallel queries to a randomly chosen dilation unitary $U_\Lambda$.
        Using the optimal isometry complex conjugation protocol constructed in this work, we can implement the complex conjugation of $V_\Lambda$ with $O(d_\mathrm{in}[\Delta_{\Lambda}/\epsilon+1])$ queries with diamond-norm error $\epsilon$.
        The optimality of isometry complex conjugation implies the lower bound on the query complexity of the channel complex conjugation given by $\Omega(d_\mathrm{out}(\min\{d_\mathrm{in}, d_\mathrm{out}, \Delta_\Lambda\}/\epsilon+1))$ for $0<\epsilon\leq 1/4$, which achieves the optimality of the former construction for constant Kraus rank $r=O(1)$.
        This construction has an advantage over the protocol based on the unitary complex conjugation~\cite{miyazaki2019complex,ebler2023optimal} and the lifting from the channel to its Stinespring dilation unitary, which requires $O(d_\mathrm{in}^2d_\mathrm{out}^2r^2/\epsilon)$ queries to $\Lambda$.}
    \label{fig:access-models}
\end{figure}

The rest of this work is organized as follows.
Section~\ref{sec:preliminaries} provides preliminaries on the Schur--Weyl duality and quantum superchannels.
Section~\ref{sec:optimal-complex-conjugation} provides the main result on the optimal fidelity of isometry complex conjugation, its circuit construction, and its application to channel complex conjugation.
Section~\ref{sec:multi-copy-conjugation} provides the extension to the multi-copy case, which shows the optimality of the estimation-based protocol in the asymptotic limit.
Section~\ref{sec:discussion} discusses the comparison with the estimation-based protocol and the advantage of isometry complex conjugation over unitary complex conjugation.
Section~\ref{sec:conclusion} concludes this work and discusses future directions.

\section{Preliminaries}
\label{sec:preliminaries}

\subsection{Schur--Weyl duality}

We summarize the Schur--Weyl duality, which is a key tool for analyzing the optimal fidelity of isometry complex conjugation.
We suggest standard textbooks~\cite{klimyk1995representations,fulton2013representation,ceccherini2010representation} for more details.
We consider the representation of the unitary group $\U(d)$ and the symmetric group $\mfS_n$ on the $n$-fold space $(\CC^d)^{\otimes n}$, given by
\begin{align}
    \U(d) \ni U \mapsto U^{\otimes n} \in \mcL((\CC^d)^{\otimes n}),\\
    \mfS_n \ni \sigma \mapsto \pi(\sigma) \in \mcL((\CC^d)^{\otimes n}),
\end{align}
where $\pi(\sigma)$ is the permutation operator defined by
\begin{align}
    \pi(\sigma)(\ket{i_1}\otimes\cdots\otimes\ket{i_n}) = \ket{i_{\sigma^{-1}(1)}}\otimes\cdots\otimes\ket{i_{\sigma^{-1}(n)}}.
\end{align}
These two representations generate the commutant of each other, which provides the decomposition of $(\CC^d)^{\otimes n}$ into irreducible representation (irrep) spaces of $\U(d)$ and $\mfS_n$:
\begin{align}
    \label{eq:schur-weyl}
    (\CC^d)^{\otimes n} &\cong \bigoplus_{\lambda\vdash_d n} \mcU_\lambda^{(d)}\otimes\mcS_\lambda,\\
    U^{\otimes n} &\cong \bigoplus_{\lambda\vdash_d n} U_\lambda\otimes\1_{\mcS_\lambda} \quad \forall U\in \U(d),\\
    \pi(\sigma) &\cong \bigoplus_{\lambda\vdash_d n} \1_{\mcU_\lambda^{(d)}}\otimes g_\lambda(\sigma) \quad \forall \sigma \in \mfS_n,
\end{align}
where $\lambda\vdash_d n$ denotes a Young diagram with $n$ boxes and at most $d$ rows, $\mcU_\lambda^{(d)}$ is the irrep space of $\U(d)$ (also called the Weyl module), $\mcS_\lambda$ is the irrep space of $\mfS_n$ (also called the Specht module), and $U_\lambda$ and $g_\lambda$ are the corresponding irreducible representations of $\U(d)$ and $\mfS_n$, respectively.
The Young diagram $\lambda$ is represented by a sequence of non-negative integers $\lambda=(\lambda_1,\ldots,\lambda_d)$ with $\lambda_1\geq \cdots \lambda_d\geq 0$.
The isomorphism above is called the quantum Schur transform $U_\mathrm{Sch}^{(n,d)}$, i.e.,
\begin{align}
    U_\mathrm{Sch}^{(n,d)}: (\CC^d)^{\otimes n} \to \bigoplus_{\lambda\vdash_d n} \mcU_\lambda^{(d)}\otimes\mcS_\lambda
\end{align}
such that
\begin{align}
    \label{eq:unitary-in-schur-basis}
    U_\mathrm{Sch}^{(n,d)} U^{\otimes n} U_\mathrm{Sch}^{(n,d)\dagger} &= \bigoplus_{\lambda\vdash_d n} U_\lambda\otimes\1_{\mcS_\lambda} \quad \forall U\in \U(d),\\
    U_\mathrm{Sch}^{(n,d)} \pi(\sigma) U_\mathrm{Sch}^{(n,d)\dagger} &= \bigoplus_{\lambda\vdash_d n} \1_{\mcU_\lambda^{(d)}}\otimes g_\lambda(\sigma) \quad \forall \sigma \in \mfS_n.
\end{align}
In particular, we take the Gelfand-Zetlin basis of $\mcU_\lambda^{(d)}$ and the Young-Yamanouchi basis of $\mcS_\lambda$, whose tensor products form the Schur basis of $(\CC^d)^{\otimes n}$.
In this case, the quantum Schur transform $U_\mathrm{Sch}^{(n,d)}$ is a real orthogonal matrix~\cite{klimyk1995representations}.

The dimension of the Weyl module $\mcU_\lambda^{(d)}$ is denoted by $\md{\lambda}$, which is calculated by the Weyl dimension formula:
\begin{theorem}[Weyl dimension formula]
For $\lambda\vdash_d n$, the dimension of the corresponding irreducible representation of $\U(d)$ is
\begin{align}
    \label{eq:unitary_irrep_dimension}
    \md{\lambda} = {\prod_{1\leq i<j\leq d}(\tilde{\lambda}_i-\tilde{\lambda}_j) \over \prod_{j=1}^{d-1}j!},
\end{align}
where $\tilde{\lambda}_i$ is defined by $\tilde{\lambda}_i \coloneqq \lambda_i-i$.
\end{theorem}

\begin{figure}
    \centering
    \includegraphics[width=0.5\linewidth]{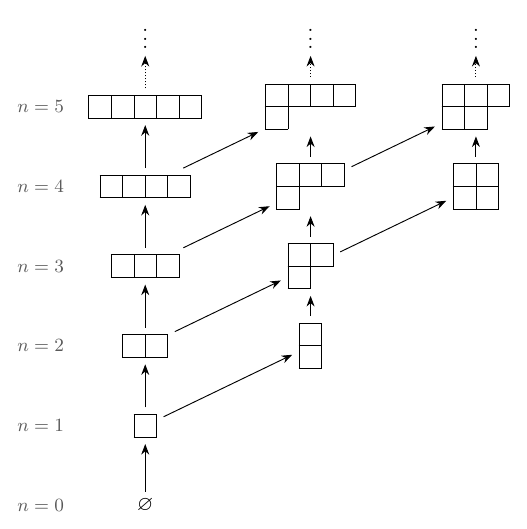}
    \caption{An example of the restricted Young lattice $\YY_d$ for $d=2$.
    It consists of vertices labeled by the Young diagrams $\lambda\vdash_d n$ for $n\in \ZZ_{\geq 0}$, and edges $\alpha\to\lambda$ with $\lambda\in \alpha+_d\square$.}
    \label{fig:young-lattice_d2}
\end{figure}

\noindent To define the basis vectors in the Specht module $\mcS_\lambda$, we introduce the restricted Young lattice $\YY_d$, which is a directed graph $\YY_d = (V_d,E_d)$ with vertex set $V_d$ and edge set $E_d$ given by (see also Fig.~\ref{fig:young-lattice_d2})
\begin{align}
    V_d &\coloneqq \{\lambda \vdash_d n \mid n\in\ZZ_{\geq 0}\},\\
    E_d &\coloneqq \{(\alpha, \lambda) \mid \lambda\in \alpha+_d \square\}.
\end{align}
where $\alpha+_d \square$ denotes the set of Young diagrams of height at most $d$ obtained from $\alpha$ by adding one box, defined by
\begin{align}
    \alpha+_d \square \coloneqq \{\lambda = \alpha+e_i \mid 1\leq i\leq d, \lambda\vdash_d \abs{\alpha}+1\},
\end{align}
where $e_i$ is the $i$-th standard basis vector given by $(e_i)_j = \delta_{ij}$ using the Kronecker delta $\delta_{ij}$ defined by $\delta_{ij} = 1$ if $i=j$ and $\delta_{ij} = 0$ otherwise.
We define the set of paths from $\alpha\in V_d$ to $\lambda \in V_d$ by
\begin{align}
    \Path{\alpha\to\lambda} \coloneqq \{\alpha_0 \to \cdots \to \alpha_k \mid k\in\ZZ_{\geq 0}, \alpha_0=\alpha, \alpha_k=\lambda, (\alpha_{i-1},\alpha_i)\in E_d\},
\end{align}
and the set of paths from the empty diagram $\varnothing$ to $\lambda$ by
\begin{align}
    \Path{\lambda} \coloneqq \Path{\varnothing\to\lambda}.
\end{align}
Then, the Specht module $\mcS_\lambda$ is spanned by basis vectors labeled by a path $p\in \Path{\lambda}$.
Thus, its dimension $d_\lambda$ is given by
\begin{align}
    \label{eq:dimension-specht}
    d_\lambda = \abs{\Path{\lambda}}.
\end{align}

By using the Schur--Weyl decomposition, we can define the irreducible matrix units $E_{pq}^{\lambda,d}$ as
\begin{align}
    E_{pq}^{\lambda,d} \coloneqq U_\mathrm{Sch}^{(n,d) \dagger}(\1_{\mcU_\lambda^{(d)}}\otimes\ketbra{p}{q}) U_\mathrm{Sch}^{(n,d)},
\end{align}
which form a basis of the commutant algebra of $\U(d)$ on $(\CC^d)^{\otimes n}$:
\begin{align}
    \Comm(\U(d), n)
    &\coloneqq \{X\in \mcL(\CC^d)^{\otimes n} \mid [X, U^{\otimes n}] = 0\}\\
    &= \mathrm{span}\{E_{pq}^{\lambda, d} \mid \lambda\vdash_d n, p, q\in \Path{\lambda}\}.
\end{align}
The irreducible matrix units satisfy the following relations~\cite{studzinski2022efficientmulti}:
\begin{theorem}
    \label{thm:branching-trace-identity}
    For any $\lambda\vdash_d n$, $\alpha, \beta\in \lambda-\square$, $p_\alpha\in \Path{\alpha}$, and $q_\beta\in\Path{\beta}$, the following branching rule holds:
\begin{align}
    \label{eq:branching-trace-identity}
    \Tr_n E^{\lambda, d}_{p_\alpha \to \lambda, q_\beta\to \lambda} = \delta_{\alpha\beta} \frac{\md{\lambda}}{\md{\alpha}} E^{\alpha}_{p_\alpha q_\alpha}.
\end{align}
\end{theorem}

\begin{theorem}
    \label{thm:branching-tensor-identity}
    For any $\alpha\vdash_d n-1$ and $p_\alpha, q_\alpha\in \Path{\alpha}$ the following branching rule holds:
\begin{align}
    E^{\alpha, d}_{p_\alpha q_\alpha}\otimes \1_d = \sum_{\lambda\in \alpha+_d\square} E^{\lambda, d}_{p_\alpha\to \lambda, q_\alpha\to \lambda}.
\end{align}
\end{theorem}

\subsection{Quantum Schur transforms and (dual) Clebsch--Gordan transforms}

Suppose we have two irreps $\alpha, \beta$ of $\U(d)$.
The tensor product $\mcU_\alpha^{(d)}\otimes \mcU_\beta^{(d)}$ can be decomposed into irreps of $\U(d)$ as
\begin{align}
    \mcU_\alpha^{(d)} \otimes \mcU_\beta^{(d)} \cong \bigoplus_{\lambda} \mcU_\lambda^{(d)} \otimes \CC^{c_{\alpha\beta}^\lambda},
\end{align}
where $c_{\alpha\beta}^\lambda$ is the multiplicity of $\mcU_\lambda^{(d)}$ in the decomposition of $\mcU_\alpha^{(d)}\otimes \mcU_\beta^{(d)}$.
The multiplicity $c_{\alpha\beta}^\lambda$ is calculated by
\begin{align}
    \label{eq:multiplicity-integral}
    c_{\alpha\beta}^\lambda = \int_{\U(d)} \Tr(U_\alpha) \Tr(U_\beta) \overline{\Tr(U_\lambda)} \dd U.
\end{align}
The above isomorphism is called the Clebsch--Gordan (CG) transform\footnote{The CG transformation usually refers to the case where $\beta$ is the fundamental representation $\beta = \square$, but we use the term CG transform in a broader sense.}, denoted by $\mathrm{CG}_{\alpha,\beta}^{(d)}: \mcU_\alpha^{(d)}\otimes \mcU_\beta^{(d)} \to \bigoplus_{\lambda\vdash_d n+m} \mcU_\lambda^{(d)}\otimes \CC^{c_{\alpha\beta}^\lambda}$.
In particular, when $\beta$ is the fundamental representation $\beta = \square$ given by $U_\square = U$ for $U\in \U(d)$, the multiplicity $c_{\alpha\square}^\lambda$ is given by
\begin{align}
    c_{\alpha\square}^\lambda
    =
    \begin{cases}
        1 & \text{if $\lambda\in \alpha+_d\square$}\\
        0 & \text{otherwise}
    \end{cases},
\end{align}
i.e.,
\begin{align}
    \mcU_\alpha^{(d)}\otimes \mcU_\square^{(d)} \cong \bigoplus_{\lambda\in \alpha+_d\square} \mcU_\lambda^{(d)},
\end{align}
and $\mathrm{CG}_{\alpha, \square}^{(d)}$ is given by
\begin{align}
    \mathrm{CG}_{\alpha, \square}^{(d)}: \mcU_\alpha^{(d)}\otimes \mcU_\square^{(d)} \to \bigoplus_{\lambda\in \alpha+_d\square} \mcU_\lambda^{(d)}
\end{align}
such that
\begin{align}
    \label{eq:CG-transform}
    \mathrm{CG}_{\alpha, \square}^{(d)}(U_\alpha \otimes U) \mathrm{CG}_{\alpha, \square}^{(d)\dagger} = \bigoplus_{\lambda\in \alpha+_d\square} U_\lambda \quad \forall U\in \U(d).
\end{align}
We define the complex conjugate $\overline{\beta}$ of an irrep $\beta$ by $U_{\overline{\beta}} = \overline{U_\beta}$ for $U_\beta\in \U(d)$.
Then, the multiplicity satisfies $c_{\alpha\beta}^\lambda = c_{\lambda\overline{\beta}}^\alpha$ due to Eq.~\eqref{eq:multiplicity-integral}.
Since $\mcU_{\overline{\beta}}^{(d)}$ is isomorphic to $\mcU_\beta^{(d)}$ as a vector space, we can take the same labeling $u_\beta\in [\md{\beta}]$ for the basis vectors of $\mcU_{\overline{\beta}}^{(d)}$ and $\mcU_\beta^{(d)}$.
In particular, we can take these basis vectors such that~\cite[Eq.~(10), p.~298]{klimyk1995representations}
\begin{align}
    \label{eq:dual-CG}
    \bra{u_\alpha} \bra{a}\mathrm{CG}_{\lambda, \overline{\beta}}^{(d)} \ket{u_\lambda} \ket{u_\beta} = \sqrt{\md{\alpha} \over \md{\lambda}} \overline{\bra{u_\lambda} \bra{a}\mathrm{CG}_{\alpha, \beta}^{(d)}\ket{u_\alpha} \ket{u_\beta}} \quad \forall u_\alpha\in [\md{\alpha}], u_\beta \in [\md{\beta}], u_\lambda \in [\md{\lambda}], a \in [c_{\alpha\beta}^\lambda].
\end{align}
We call the isomorphism $\mathrm{CG}_{\lambda, \overline{\beta}}^{(d)}$ the dual CG transform\footnote{The dual CG transform usually refers to $\mathrm{CG}_{\lambda, \overline{\square}}$, but we use the term in a broader sense.}.
By concatenating the CG transforms $\mathrm{CG}_{\alpha, \square}^{(d)}$ for irreps $\alpha$, we can construct the quantum Schur transforms $U_\mathrm{Sch}^{(n,d)}$ as follows:
\begin{align}
    (\CC^{d})^{\otimes n} &\cong \mcU_\square^{(d)} \otimes \cdots \otimes \mcU_\square^{(d)}\\
    &\cong \bigoplus_{\alpha\vdash_d 2} \mcU_\alpha^{(d)} \otimes \mcU_\square^{(d)} \otimes \cdots \mcU_\square^{(d)}\\
    &\cong \cdots \\
    &\cong \bigoplus_{\lambda\vdash_d n} \mcU_\lambda^{(d)} \otimes \mathrm{span}\{\ket{p} \mid p\in \Path{\lambda}\}.
\end{align}
The multiplicity space $\mathrm{span}\{\ket{p} \mid p\in \Path{\lambda}\}$ is isomorphic to the Specht module $\mcS_\lambda$ due to the Schur--Weyl duality shown in Eq.~\eqref{eq:schur-weyl}.

\begin{figure}[t]
    \centering
    \includegraphics{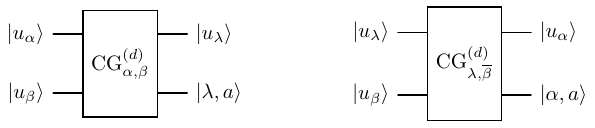}
    \caption{(Left panel) The CG transforms in the quantum circuit.
    Two input registers correspond to the Weyl modules $\mcU_\alpha^{(d)}$ and $\mcU_\beta^{(d)}$.
    The first output register corresponds to the Weyl module $\mcU_\lambda^{(d)}$.
    The second output register contains the irrep label $\lambda$ and the multiplicity label $a\in[c_{\alpha\beta}^{\lambda}]$.
    (Right panel) The dual CG transforms in the quantum circuit, similar to the left panel.}
    \label{fig:cg-transforms}
\end{figure}

\begin{figure}[t]
    \centering
    \includegraphics{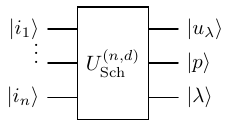}
    \caption{The quantum Schur transform $U_\mathrm{Sch}^{(n,d)}$ in the quantum circuit.
    It maps the computational-basis registers to the Weyl-module register $\ket{u_\lambda}$, the Specht-module register $\ket{p}$, and the Young-diagram register $\ket{\lambda}$.}
    \label{fig:quantum-schur-transform}
\end{figure}

The (dual) CG transforms and the quantum Schur transform can be represented as quantum circuits as shown in Figs.~\ref{fig:cg-transforms} and~\ref{fig:quantum-schur-transform}.
The CG transforms $\mathrm{CG}_{\alpha, \square}^{(d)}$ and the quantum Schur transform  $U_\mathrm{Sch}^{(n,d)}$ can be implemented with the circuit complexity $O(\mathrm{poly}(\log d, n))$~\cite{burchardt2025high} and the dual CG transform $\mathrm{CG}_{\alpha, \overline{\square}}^{(d)}$ can be implemented with the circuit complexity $O(\mathrm{poly}(d, n))$~\cite{nguyen2023mixed,grinko2023gelfand,fei2023efficient,grinko2025mixed}.
The quantum circuit for the CG transforms $\mathrm{CG}_{\alpha, \beta}^{(d)}$ for general irreps $\alpha$ and $\beta$ is not known.

\subsection{Stiefel manifold, its Haar measure, and Schur--Weyl duality applied to isometry channels}

The set of isometry operators $V: \CC^d\to \CC^D$ can be identified with the Stiefel manifold $\isometry{d}{D}$, which is the set of $D\times d$ complex matrices with orthonormal columns~\cite{james1976topology}, i.e.,
\begin{align}
    \isometry{d}{D} = \{V: \CC^d \to \CC^D \mid V^\dagger V = \1_d\}.
\end{align}
It is isomorphic to the quotient space $\U(D)/\U(D-d)$, and the Haar measure on $\U(D)$ induces a unique normalized invariant measure $\dd V$ on $\isometry{d}{D}$, which satisfies the left- and right-invariance property $\dd (WVU) = \dd V$ for all $U\in\U(d)$ and $W\in\U(D)$.

The $n$-fold unitary operator $U^{\otimes n}$ for $U\in \U(d)$ is represented in the Schur basis as Eq.~\eqref{eq:unitary-in-schur-basis}, and the tensor product of an irrep $U_\alpha$ of $\U(d)$ and the fundamental representation $U$ is represented as Eq.~\eqref{eq:CG-transform}.
These representations can be extended to isometry operators $V:\CC^d\to\CC^D$ as follows~\cite{yoshida2023universal,yoshida2025universal,yoshida2026quantum}:
\begin{align}
    U_{\mathrm{Sch}}^{(n,D)}V^{\otimes n}
    U_{\mathrm{Sch}}^{(n,d)\dagger}
    &=
    \bigoplus_{\alpha\vdash_d n}V_\alpha\otimes\1_{\mcS_\alpha} \quad \forall V\in \isometry{d}{D},
    \label{eq:schur-isometry}\\
    \mathrm{CG}_{\alpha, \square}^{(D)}(V_\alpha \otimes V) \mathrm{CG}_{\alpha, \square}^{(d)\dagger} &= \bigoplus_{\lambda\in \alpha+_d\square} V_\lambda \quad \forall V\in \isometry{d}{D},
\end{align}
where $V_\alpha: \mcU_\alpha^{(d)} \to \mcU_\alpha^{(D)}$ is an isometry between the Weyl modules of $\U(d)$ and $\U(D)$ associated with the same Young diagram $\alpha$.
The map $V\mapsto V_\alpha$ is not a representation, but it admits a similar expression: $(WVU)_\alpha = W_\alpha V_\alpha U_\alpha$ holds for $V\in \isometry{d}{D}$, $U\in \U(d)$, and $W\in \U(D)$.

We introduce the following lemma for the calculation of the Haar integral over $\isometry{d}{D}$:

\begin{lemma}
    \label{lem:haar-integral}
    For any $\alpha, \beta\vdash_d n$, the following relation holds:
    \begin{align}
        \int_{\isometry{d}{D}} \dd V V_\alpha (\cdot) V_\beta^\dagger = \delta_{\alpha\beta} {\1_{\mcU_\alpha^{(D)}} \over \mD{\alpha}} \Tr(\cdot).
    \end{align}
\end{lemma}

\begin{proof}
    Due to the left-invariance $\dd(WV) = \dd V$ of the Haar measure $\dd V$ for all $W\in \U(D)$, the left-hand side satisfies
    \begin{align}
        W_\alpha \left(\int_{\isometry{d}{D}} \dd V V_\alpha (\cdot) V_\beta^\dagger\right) W_\beta^\dagger = \int_{\isometry{d}{D}} \dd V V_\alpha (\cdot) V_\beta^\dagger
    \end{align}
    for all $W\in \U(D)$.
    Due to Schur's lemma, we obtain
    \begin{align}
        \int_{\isometry{d}{D}} \dd V V_\alpha (\cdot) V_\beta^\dagger
        \propto
        \begin{cases}
            0 & (\alpha\neq \beta)\\
            \1_{\mcU_\alpha^{(D)}} & (\alpha = \beta)
        \end{cases}.
    \end{align}
    By comparing the traces of both sides, we obtain the desired equality.
\end{proof}

\subsection{Semidefinite programming for optimal quantum superchannels}

In this work, we consider the task of transforming an unknown isometry channel $V:\CC^d\to\CC^D$ into its complex conjugate channel $\overline{V}:\CC^d\to\CC^D$.
We first summarize the notion of quantum superchannels describing the transformation of quantum channels (see Ref.~\cite{taranto2025higher} for the review).
Then, we define the fidelity-based figure of merit for the task and provide the semidefinite programming (SDP) formulation for the optimal fidelity based on Ref.~\cite{chiribella2016optimal}.

The quantum superchannel is a multi-linear map transforming $n$ input channels $\Phi_1,\ldots,\Phi_n$ for $\Phi_i: \mcL(\mcI_i)\to \mcL(\mcO_i)$, $i\in [n]$ into an output channel $\Phi_\mathrm{out}: \mcL(\mcP)\to \mcL(\mcF)$, where $\mcI_i$ and $\mcO_i$ are the input and output Hilbert spaces of the $i$-th channel, $\mcP$ and $\mcF$ are the input and output Hilbert spaces of the output channel, and $\mcL(\mcX)$ is the set of linear operators on a Hilbert space $\mcX$.
It is naturally represented as a linear map transforming $\Phi_1\otimes \cdots \otimes \Phi_n$ to $\Phi_\mathrm{out}$ as
\begin{align}
    \sfS: \bigotimes_{i=1}^{n} [\mcL(\mcI_i)\to \mcL(\mcO_i)] \to [\mcL(\mcP)\to \mcL(\mcF)],
\end{align}
where $[\mcL(\mcX) \to \mcL(\mcY)]$ represents the set of linear maps from $\mcL(\mcX)$ to $\mcL(\mcY)$.
We can represent the input and output channels by their Choi matrices $J_{\Phi_i}\in \mcL(\mcI_i\otimes \mcO_i)$ and $J_{\Phi_\mathrm{out}}\in \mcL(\mcP\otimes \mcF)$, where the Choi matrix $J_{\Phi}$ of $\Phi: \mcL(\mcI)\to \mcL(\mcO)$ is defined by
\begin{align}
    J_\Phi &\coloneqq (\1_{\mcL(\mcI)}\otimes \Phi)(\dketbra{\1}_\mcI) \in \mcL(\mcI\otimes \mcO),\\
    \dket{\1}_{\mcI}&\coloneqq \sum_{i=1}^{\dim(\mcI)} \ket{i}\otimes \ket{i} \in \mcI\otimes \mcI,
\end{align}
using the computational basis $\{\ket{i}\}$ of $\mcI$.
Then, the superchannel $\sfS$ is represented by a matrix $J_{\sfS}\in \mcL(\mcI^n\otimes \mcO^n\otimes \mcP\otimes \mcF)$ called the Choi matrix of the superchannel $\sfS$, where $\mcI^n$ and $\mcO^n$ are joint Hilbert spaces defined by $\mcI^n \coloneqq \bigotimes_{i=1}^{n} \mcI_i$ and $\mcO^n \coloneqq \bigotimes_{i=1}^{n} \mcO_i$.
The Choi matrix $J_{\sfS}$ satisfies
\begin{align}
    J_{\Phi_\mathrm{out}} = J_{\sfS} \ast \bigotimes_{i=1}^{n} J_{\Phi_i},
\end{align}
where $\ast$ is the link product defined by
\begin{align}
    A \ast B \coloneqq \Tr_{\mcX}[(A\otimes \1_{\mcZ})(B^{\top_\mcX} \otimes \1_{\mcY})],
\end{align}
for $A\in \mcL(\mcX\otimes \mcY)$ and $B\in \mcL(\mcX\otimes \mcZ)$, and $\top_\mcX$ is the partial transpose on $\mcX$.
The link product of the Choi matrices of two quantum channels represents the composition of the channels, i.e., for $\Phi_1: \mcL(\mcV)\to \mcL(\mcW \otimes \mcX)$ and $\Phi_2: \mcL(\mcX \otimes \mcY)\to \mcL(\mcZ)$, we have
\begin{align}
    J_{\Phi_1} \ast J_{\Phi_2} = J_{(\Phi_2 \otimes \1_{\mcL(\mcW)}) \circ (\1_{\mcL(\mcY)} \otimes \Phi_1)}.
\end{align}

We define the parallel quantum superchannel, which is realized in a quantum circuit by calling the input channels $\Phi_1, \ldots, \Phi_n$ in parallel, as
\begin{align}
    \sfS\left(\bigotimes_{i=1}^{n} \Phi_i\right) = \mcD \circ \left(\bigotimes_{i=1}^{n} \Phi_i\right) \circ \mcE,
\end{align}
using an encoder channel $\mcE: \mcL(\mcP)\to \bigotimes_{i=1}^{n} \mcL(\mcI_i) \otimes \mcL(\mcA)$ and a decoder channel $\mcD: \bigotimes_{i=1}^{n} \mcL(\mcO_i) \otimes \mcL(\mcA) \to \mcL(\mcF)$, where $\mcA$ is an auxiliary Hilbert space.
Then, the Choi matrix $J_{\sfS}$ of a parallel superchannel is represented as $J_{\sfS} = J_{\mcD} \ast J_{\mcE}$, and such a Choi matrix $J_{\sfS}$ is characterized by the following conditions~\cite{chiribella2008quantumcircuits}:
\begin{align}
    J_{\sfS} &\geq 0,\\
    \Tr_{\mcF}J_{\sfS} &= C \otimes \1_{\mcO^n},\\
    \Tr_{\mcI^n}C &= \1_{\mcP},
\end{align}
where $C \coloneqq \Tr_{\mcF \mcO^n}J_{\sfS}/\dim(\mcO^n)$.
By introducing the linear map $\mcL_\mathrm{PAR}(J_\sfS)\coloneqq (\Tr_{\mcF}J_{\sfS} - C \otimes \1_{\mcO^n}) \oplus (\Tr_{\mcI^n}C - \1_{\mcP})$, the above condition is summarized as
\begin{align}
    J_\sfS\geq 0, \quad \mcL_{\mathrm{PAR}}(J_\sfS) = 0.
\end{align}
The most general quantum superchannel is defined in an axiomatic way, which is defined as a linear map $\sfS$ satisfying the following conditions:
\begin{itemize}
    \item Completely CP preserving: For any $m_1, \ldots, m_n\in \ZZ_{\geq 0}$ and any completely positive (CP) maps $\Phi_i: \mcL(\mcI_i) \to \mcL(\mcO_i \otimes \CC^{m_i})$ for $i\in [n]$, the map $(\sfS \otimes \1_{\mcL(\CC^{m_1\cdots m_n})})(\Phi_1\otimes \cdots \Phi_n)$ is also CP.
    \item TP preserving: For any trace preserving (TP) maps $\Phi_1, \ldots, \Phi_n$, the output channel $\sfS(\Phi_1 \otimes \cdots \otimes \Phi_n)$ is also TP.
\end{itemize}
The above conditions are represented in terms of the Choi matrix $J_{\sfS}$ as follows:
\begin{align}
    J_{\sfS} \geq 0, \quad \mcL_{\mathrm{GEN}}(J_{\sfS}) =0,
\end{align}
where $\mcL_{\mathrm{GEN}}$ is a certain linear map (see Ref.~\cite{oreshkov2016causal} for the explicit form).
We define the set of parallel and general quantum superchannels as
\begin{align}
    X \coloneqq \{\sfS\mid J_{\sfS}\geq 0, \mcL_X(J_\sfS)=0\}
\end{align}
for $X\in \{\mathrm{PAR}, \mathrm{GEN}\}$.
Any quantum superchannel implementable in the quantum circuit satisfies the above two conditions, such as the parallel one defined above, the sequential one (also called the quantum comb), causally separable one, and quantum circuits with classical control (QC-CC).
It also includes the quantum superchannels beyond the quantum circuit model, called the quantum superchannels with indefinite causal order~\cite{hardy2007towards,oreshkov2012quantum,chiribella2013quantum}.
Thus, the set of general quantum superchannels can be considered as an outer approximation of the set of quantum superchannels implementable in the quantum circuit model.
In this work, we consider the general quantum superchannels to obtain the upper bound of the optimal fidelity for transforming an unknown isometry channel into its complex conjugate channel.

In this work, we consider the task of transforming an unknown isometry channel $V:\CC^d\to\CC^D$ into another isometry channel $f(V)$, where $f$ is a function on $\isometry{d}{D}$.
To quantify the performance, we introduce the average-case fidelity defined by
\begin{align}
    F_{\mathrm{avg}}(\sfS, f)
    \coloneqq \int_{\isometry{d}{D}} \dd V F_{\mathrm{ch}}(\sfS(\mcV^{\otimes n}), f(\mcV)),
\end{align}
where $\mcV(\cdot)\coloneqq V(\cdot)V^\dagger$ is the isometry channel corresponding to $V\in \isometry{d}{D}$, $f(\mcV)$ is the isometry channel corresponding to $f(V)$, and $F_{\mathrm{ch}}$ is the channel fidelity defined by
\begin{align}
    F_{\mathrm{ch}}(\Phi,\mcV)
    \coloneqq
    \frac{1}{d^2}\dbra{V}J_\Phi\dket{V},
    \label{eq:channel-fidelity}
\end{align}
for a CPTP map $\Phi$ and an isometry channel $\mcV$ using the dual ket $\dket{V}$ defined by
\begin{align}
    \dket{V}\coloneqq \sum_{i=1}^{d} \ket{i}\otimes V\ket{i} \in \CC^d\otimes \CC^D.
\end{align}
Then, the average-case fidelity is expressed as a linear function of $J_\sfS$ by
\begin{align}
    F_{\mathrm{avg}}(\sfS, f) = \Tr(J_\sfS \Omega_{f, n})
\end{align}
using the performance operator $\Omega_{f, n}$ defined by
\begin{align}
    \label{eq:performance-operator}
    \Omega_{f, n} \coloneqq \frac{1}{d^2} \int_{\isometry{d}{D}} \dd V \dketbra{f(V)}_{\mcP\mcF} \otimes \overline{\dketbra{V}}_{\mcI^n\mcO^n}^{\otimes n}.
\end{align}
For the tasks considered in this work, we can show that the optimal average-case fidelity is the same as the optimal worst-case fidelity, which is defined by
\begin{align}
    F_{\mathrm{worst}}(\sfS, f) \coloneqq \min_{V\in \isometry{d}{D}} F_{\mathrm{ch}}(\sfS(\mcV^{\otimes n}), f(\mcV)),
\end{align}
due to the unitary-group covariance [see also Eq.~\eqref{eq:unitary-symmetry}].

The optimization of the performance in the parallel or general quantum superchannel is formalized as the following SDP:
\begin{align}
    \begin{split}
        \label{eq:sdp}
    \max\; &\Tr(J_\sfS \Omega_{f, n})\\
    \mathrm{s.t.} \; & J_{\sfS}\geq 0, \quad \mcL_X (J_{\sfS}) = 0,
    \end{split}
\end{align}
for $X\in \{\mathrm{PAR}, \mathrm{GEN}\}$.
In particular, for the case of $X=\mathrm{PAR}$, this SDP is given by
\begin{align}
    \begin{split}
        \label{eq:primal-sdp-parallel}
    \max\; &\Tr(J_\sfS \Omega_{f, n})\\
    \mathrm{s.t.} \; & J_{\sfS}\geq 0,\\
    &\Tr_{\mcF}J_{\sfS} = C \otimes \1_{\mcO^n},\\
    &\Tr_{\mcI^n}C = \1_{\mcP}.
    \end{split}
\end{align}
We can consider the dual problem of the SDP~\eqref{eq:sdp} defined by
\begin{align}
    \begin{split}
        \label{eq:dual-sdp}
    \min\; & c\\
    \mathrm{s.t.}\; & \Omega_{f,n} \leq c W, \quad W\in L_X^*,
    \end{split}
\end{align}
where $L_X$ is an affine space defined by $L_X\coloneqq \{A\mid \mcL_X(A) = 0\}$ and $L_X^*$ is its affine dual defined by $L_X^*\coloneqq \{B\mid \Tr(AB^\dagger) = 1 \quad \forall A\in L_X\}$.
Due to the weak duality, any feasible solution of the dual SDP~\eqref{eq:dual-sdp-general} provides an upper bound of the primal SDP~\eqref{eq:sdp}.
For the case of $X=\mathrm{GEN}$, the dual SDP is given by~\cite{bavaresco2021strict}
\begin{align}
    \begin{split}
        \label{eq:dual-sdp-general}
    \min\; & c\\
    \mathrm{s.t.}\; & \Omega_{f,n} \leq c (W\otimes \1_\mcF),\\
    & \Tr_{\mcO_i}W
    =\Tr_{\mcI_i\mcO_i}W\otimes\frac{\1_{\mcI_i}}{d} \quad (i\in [n]),\\
    & \Tr W=d^n.
    \end{split}
\end{align}

\subsection{Lemmas in matrix inequality}

We introduce the following lemma for the matrix inequality, which is used in the proof of the optimality of the parallel protocol for complex conjugation of an unknown isometry channel.

\begin{lemma}
    \label{eq:matrix-inequality}
    Suppose $V$ is a finite-dimensional complex vector space, $\{\ket{a}\}_{a\in A}\subset V$ is a set of vectors labeled by a finite set $A$, and $\{c_a\}_{a\in A}\subset \CC$ is a set of complex coefficients such that $c_a$ are not all zero.
    Defining $\ket{b}\coloneqq\sum_{a\in A}c_a\ket{a}$, we have
    \begin{align}
        \sum_{a\in A}\ketbra{a}
        \geq
        \frac{\ketbra{b}}{\sum_{a\in A}|c_a|^2}.
    \end{align}
\end{lemma}
\begin{proof}
    It is enough to test the operator inequality on an arbitrary vector $\ket{\psi}\in V$.
    Due to the Cauchy--Schwarz inequality, we have
    \begin{align}
        \braket{\psi}{b} \braket{b}{\psi}
        &= \abs{\sum_{a\in A} c_a \braket{\psi}{a}}^2\\
        &\leq \left(\sum_{a\in A} \abs{c_a}^2 \right) \left(\sum_{a\in A} \abs{\braket{\psi}{a}}^2\right)
    \end{align}
    Thus, we have
    \begin{align}
        \frac{\braket{\psi}{b}\braket{b}{\psi}}
        {\sum_{a\in A}|c_a|^2}
        &\leq \sum_{a\in A} \abs{\braket{\psi}{a}}^2\\
        &= \bra{\psi} \sum_{a\in A} \ketbra{a} \ket{\psi}.
    \end{align}
\end{proof}

\section{Optimal complex conjugation of isometry channels}
\label{sec:optimal-complex-conjugation}

Here we show the following theorem on the optimal fidelity for complex conjugation of an unknown isometry channel from $n$ calls.
We define the optimal fidelity of isometry complex conjugation $f: V\mapsto \overline{V}$ for $V\in\isometry{d}{D}$ in the parallel and general quantum superchannels by
\begin{align}
    F_{d,D,n}^{(X)} \coloneqq \sup_{\sfS\in X} F_\mathrm{avg}(\sfS, f)
\end{align}
for $X\in \{\mathrm{PAR}, \mathrm{GEN}\}$.
Note that this result is consistent with the known result for unitary complex conjugation when $D=d$~\cite{miyazaki2019complex,ebler2023optimal}, and for state transposition when $d=1$~\cite{buzek1999optimal,brzic2026optimal}.

\begin{figure}
    \centering
    \includegraphics{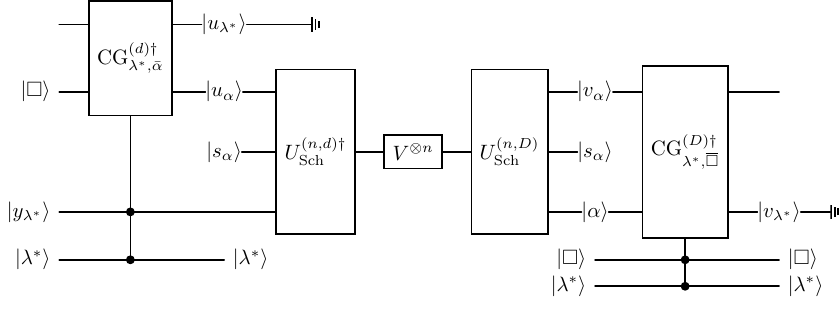}
    \caption{The optimal parallel protocol for complex conjugation of an unknown isometry $V:\CC^d\to\CC^D$ from $n$ calls.
    See the proof sketch of Prop.~\ref{prop:optimal-circuit} for the details of the implementation.}
    \label{fig:optimal-circuit}
\end{figure}

\begin{theorem}
    \label{thm:optimal-fidelity}
The optimal average channel fidelity for complex conjugation of an unknown isometry $V:\CC^d\to\CC^D$ from $n$ calls is given by
\begin{align}
F_{d,D,n}^{(\mathrm{PAR})}
&=F_{d,D,n}^{(\mathrm{GEN})}
=\max_{\lambda\vdash_d n+1}
{1\over d}
\frac{\sum_{\alpha\in\lambda-\square}\mD{\alpha}}
{\mD{\lambda}}
\label{eq:max_form}\\
&=
\begin{cases}
\dfrac{n+1}{d(D-n)} & (q=0) \\
\dfrac{1}{d}\left[D-(D-d)R_{d,D,n}\right]
& (q\geq 1)
\end{cases},
\label{eq:closed-form}
\end{align}
where $R_{d,D,n}$ is defined by
\begin{align}
    R_{d,D,n}
    &\coloneqq
    \frac{(q+D+1)(q+D-r)}
    {(q+D-r+1)(q+D-d)},
    \label{eq:dimension-ratio}
\end{align}
using $q\in \ZZ_{\geq 0}$ and $0\leq r<d$ such that $n+1 = qd+r$.
The balanced diagram
\begin{align}
    \lambda^*=\bigl((q+1)^r,q^{d-r}\bigr)\coloneqq (\underbrace{q+1,\ldots,q+1}_{r},\underbrace{q,\ldots,q}_{d-r})
    \label{eq:balanced-diagram}
\end{align}
is a maximizer in Eq.~\eqref{eq:max_form}, and the parallel protocol in Fig.~\ref{fig:optimal-circuit} attains the optimal fidelity.
\end{theorem}
\begin{proof}
    We leave the proof of Eq.~\eqref{eq:max_form} to Props.~\ref{prop:achievability} and \ref{prop:optimality}, and leave the circuit implementation to Prop.~\ref{prop:optimal-circuit}.
    Here, we maximize its Young-diagram ratio and derive Eq.~\eqref{eq:closed-form}.
    Defining $S_D(\lambda)$ for $\lambda \vdash_d n+1$ by
    \begin{align}
        S_D(\lambda)
        \coloneqq\sum_{\alpha\in\lambda-\square} \frac{\mD{\alpha}}{\mD{\lambda}},
    \end{align}
    Eq.~\eqref{eq:max_form} yields
    \begin{align}
        F_{d,D,n}^{(\mathrm{PAR})}
        =F_{d,D,n}^{(\mathrm{GEN})}
        ={1\over d}\max_{\lambda\vdash_d n+1} S_D(\lambda).
    \end{align}
    Thus, we consider the maximization of $S_D(\lambda)$ over $\lambda\vdash_d n+1$.
    Due to the Weyl dimension formula~\eqref{eq:unitary_irrep_dimension}, we have
    \begin{align}
        \mD{\lambda} = \frac{\displaystyle\prod_{1\leq i<j\leq h} (\tilde{\lambda}_i - \tilde{\lambda}_j) \cdot \prod_{1\leq i\leq h} \frac{(\tilde{\lambda}_i+D)!}{(\tilde{\lambda}_i+h)!}}{\displaystyle\prod_{k=D-h}^{D-1} k!},
    \end{align}
    where $h=H(\lambda)$ is the height of $\lambda$.
    Thus, for $1\leq i\leq h$ satisfying $\lambda-e_i\vdash_d n$, we have
    \begin{align}
        \label{eq:branching_ratio}
        \frac{\mD{\lambda-e_i}}{\mD{\lambda}}
        =\frac{\tilde{\lambda}_i+h}{\tilde{\lambda}_i+D}
        \prod_{\substack{j=1\\j\neq i}}^h
        \frac{\tilde{\lambda}_i-\tilde{\lambda}_j-1}{\tilde{\lambda}_i-\tilde{\lambda}_j}.
    \end{align}
    If $\lambda-e_i$ is not a valid Young diagram, then $\tilde{\lambda}_i-\tilde{\lambda}_j-1=0$ for some $j\neq i$, and the right-hand side of Eq.~\eqref{eq:branching_ratio} vanishes.
    Thus, $S_D(\lambda)$ defined by 
    \begin{align}
    S_D(\lambda)
    \coloneqq\sum_{\alpha\in\lambda-\square} \frac{\mD{\alpha}}{\mD{\lambda}}
    \end{align}
    is given by
    \begin{align}
        S_D(\lambda) = \sum_{i=1}^h
        \frac{\tilde{\lambda}_i+h}{\tilde{\lambda}_i+D}
        \prod_{\substack{j=1\\j\neq i}}^h
        \frac{\tilde{\lambda}_i-\tilde{\lambda}_j-1}{\tilde{\lambda}_i-\tilde{\lambda}_j}.
    \end{align}
    We consider a meromorphic function $f$ on $\CC$ defined by
    \begin{align}
        f(z)\coloneqq\frac{z+h}{z+D} \prod_{j=1}^h\frac{z-\tilde{\lambda}_j-1}{z-\tilde{\lambda}_j},
    \end{align}
    which has singular points $z=-D, \tilde{\lambda}_1, \ldots, \tilde{\lambda}_h$.
    Each of these singular points is either a simple pole or a removable singularity since $-D<\lambda_h<\cdots <\lambda_1$ holds.
    The corresponding residues are given by
    \begin{align}
        \mathop{\mathrm{Res}}_{z=\tilde{\lambda}_i}f(z)
        &=-\frac{\tilde{\lambda}_i+h}{\tilde{\lambda}_i+D}
        \prod_{\substack{j=1\\j\neq i}}^h
        \frac{\tilde{\lambda}_i-\tilde{\lambda}_j-1}{\tilde{\lambda}_i-\tilde{\lambda}_j},\\
        \mathop{\mathrm{Res}}_{z=-D}f(z)
        &=-(D-h)
        \prod_{j=1}^h\frac{D+\tilde{\lambda}_j+1}{D+\tilde{\lambda}_j}.
    \end{align}
    Since $f(z)=1-D/z+O(\abs{z}^{-2})$ as $\abs{z}\to\infty$ holds, the residue theorem yields
    \begin{align}
        \sum_{i=1}^{h} \mathop{\mathrm{Res}}_{z=\tilde{\lambda}_i}f(z) + \mathop{\mathrm{Res}}_{z=-D}f(z) = -D,
    \end{align}
    i.e.,
    \begin{align}
        S_D(\lambda)
        &=-\sum_{i=1}^{h} \mathop{\mathrm{Res}}_{z=\tilde{\lambda}_i}f(z)\\
        &= D+\mathop{\mathrm{Res}}_{z=-D}f(z)\\
        &=D-(D-h)
        \prod_{j=1}^h
        \frac{D+\lambda_j-j+1}{D+\lambda_j-j}.
        \label{eq:residue-evaluation}
    \end{align}

    We next show that the balanced diagram maximizes Eq.~\eqref{eq:residue-evaluation}.
    First suppose $D>d$.
    Then, since $\lambda_j=0$ for $j>h$, we have
    \begin{align}
        \prod_{j=h+1}^{d} \frac{D+\lambda_j-j+1}{D+\lambda_j-j} = \frac{D-h}{D-d},
    \end{align}
    which transforms Eq.~\eqref{eq:residue-evaluation} into
    \begin{align}
        S_D(\lambda)
        &=D-(D-d)P_D(\lambda),\\
        P_D(\lambda)&\coloneqq \prod_{j=1}^{d}
        \frac{D+\lambda_j-j+1}{D+\lambda_j-j}.
    \end{align}
    If $\lambda$ is not balanced, we can take $i, j$ such that $\lambda_i-\lambda_j\geq 2$ and $i<j$ holds.
    Then $\lambda'\coloneqq \lambda-e_i+e_j$ is a Young diagram of size $n+1$.
   Defining $x=D+\lambda_i-i$ and $y=D+\lambda_j-j+1$, we have $x> y>1$ and
    \begin{align}
        \frac{P_D(\lambda')}{P_D(\lambda)}
        =\frac{x^2}{x^2-1}\frac{y^2-1}{y^2} < 1.
        \label{eq:balancing-transfer}
    \end{align}
    Thus, $S_D(\lambda')> S_D(\lambda)$ holds, which shows that the balanced diagram $\lambda^*$ in Eq.~\eqref{eq:balanced-diagram} is the unique maximizer in Eq.~\eqref{eq:max_form}.
    For the case of $D=d$, the same argument shows that $\lambda^*$ is a maximizer, but it need not be unique for the case of $n\geq d$.
    In particular, every diagram of height $d$ gives the maximum value in Eq.~\eqref{eq:max_form}, which is $1$.
    Finally, substituting $\lambda^*$ into Eq.~\eqref{eq:residue-evaluation} concludes Eq.~\eqref{eq:closed-form}.
\end{proof}

In the following subsections, we give the achievability and optimality to show Eq.~\eqref{eq:max_form} in Thm.~\ref{thm:optimal-fidelity}, with the quantum circuit implementation.

\subsection{Evaluation of the performance operator}

We denote the performance operator $\Omega_{f,n}$ for the function $f(V) = \overline{V}$ for $V\in \isometry{d}{D}$ by $\Omega_{d,D,n}$, and evaluate its form.
By definition~\eqref{eq:performance-operator}, it is given by
\begin{align}
    \Omega_{d,D,n}
    &= {1\over d^2}\int_{\isometry{d}{D}} \dd V \dketbra{\overline{V}}_{\mcP\mcF} \otimes \overline{\dketbra{V}}_{\mcI^n \mcO^n}^{\otimes n}\\
    &= {1\over d^2}\int_{\isometry{d}{D}} \dd V \dketbra{V}^{\otimes n+1}_{\mcI^n\mcP, \mcO^n\mcF}\\
    &= {1\over d^2}\int_{\isometry{d}{D}} \dd V \bigoplus_{\lambda, \mu\vdash_d n+1} \dketbra{U_\mathrm{Sch}^{(n+1,D)\dagger}(V_\lambda\otimes \1_{\mcS_\lambda})U_\mathrm{Sch}^{(n+1,d)}}{U_\mathrm{Sch}^{(n+1,D)\dagger}(V_\mu \otimes \1_{\mcS_\mu})U_\mathrm{Sch}^{(n+1,d)}}\\
    &= {1\over d^2}\int_{\isometry{d}{D}} \dd V \bigoplus_{\lambda, \mu\vdash_d n+1} (U_\mathrm{Sch}^{(n+1,d)\top} \otimes U_\mathrm{Sch}^{(n+1, D)\dagger})\dketbra{V_\lambda\otimes \1_{\mcS_\lambda}}{V_\mu\otimes \1_{\mcS_\mu}}(\overline{U_\mathrm{Sch}^{(n+1,d)}} \otimes U_\mathrm{Sch}^{(n+1, D)})\\
    &= {1\over d^2}\int_{\isometry{d}{D}} \dd V \bigoplus_{\lambda, \mu\vdash_d n+1} (U_\mathrm{Sch}^{(n+1,d)\dagger} \otimes U_\mathrm{Sch}^{(n+1, D)\dagger})\dketbra{V_\lambda\otimes \1_{\mcS_\lambda}}{V_\mu\otimes \1_{\mcS_\mu}}(U_\mathrm{Sch}^{(n+1,d)} \otimes U_\mathrm{Sch}^{(n+1, D)})\\
    &= {1\over d^2}\bigoplus_{\lambda \vdash_d n+1} (U_\mathrm{Sch}^{(n+1,d)\dagger} \otimes U_\mathrm{Sch}^{(n+1, D)\dagger})\left[{\1_{\mcU_\lambda^{(d)}} \otimes \1_{\mcU_\lambda^{(D)}} \over \mD{\lambda}} \otimes \dketbra{\1_{\mcS_\lambda}}\right](U_\mathrm{Sch}^{(n+1,d)} \otimes U_\mathrm{Sch}^{(n+1, D)})\\
    &= \sum_{\lambda\vdash_d n+1} \sum_{p_\lambda,q_\lambda\in\Path{\lambda}}{(E^{\lambda, d}_{p_\lambda q_\lambda})_{\mcI^n\mcP} \otimes (E^{\lambda, D}_{p_\lambda q_\lambda})_{\mcO^n\mcF} \over d^2 \mD{\lambda}},
    \label{eq:performance-operator-schur}
\end{align}
where we change the integral variable $V\mapsto \overline{V}$ in the second equality, we use Eq.~\eqref{eq:schur-isometry} in the third equality, we use the transpose trick $\dket{ABC} = (C^{\top} \otimes A) \dket{B}$ using the transpose $C^{\top}$ of $C$ in the fourth equality, we use the reality of $U_\mathrm{Sch}^{(n+1,d)}$ in the fifth equality, we use Lem.~\ref{lem:haar-integral} in the sixth equality.

\subsection{Achievability}

We construct a feasible solution of the primal SDP~\eqref{eq:primal-sdp-parallel} to show the achievability of the optimal fidelity in Eq.~\eqref{eq:max_form}.

\begin{proposition}
    \label{prop:achievability}
The optimal fidelity for complex conjugation of an unknown isometry $V:\CC^d\to\CC^D$ from $n$ calls with parallel protocols is lower bounded by
\begin{align}
    F_{d,D,n}^{(\mathrm{PAR})}
    \geq {1\over d}\max_{\lambda\vdash_d n+1}
    \frac{\sum_{\alpha\in\lambda-\square}\mD{\alpha}}
    {\mD{\lambda}}.
\end{align}
\end{proposition}
\begin{proof}
    We fix a diagram $\lambda\vdash_d n+1$ attaining the maximum in Eq.~\eqref{eq:max_form}, and choose one path $p_\alpha\in\Path{\alpha}$ for every $\alpha\in\lambda-\square$.
    We define
    \begin{align}
        M_\lambda\coloneqq\sum_{\alpha\in\lambda-\square}\mD{\alpha},
        \label{eq:definition_M}
    \end{align}
    and
    \begin{align}
        C_1&\coloneqq \frac{1}{M_{\lambda}}
        \sum_{\alpha,\beta\in\lambda-\square}
        \frac{d\mD{\alpha}\mD{\beta}}
        {\md{\lambda}\mD{\lambda}} 
        (E^{\lambda}_{p_\alpha\to\lambda,p_\beta\to\lambda})_{\mcI^n\mcP} \otimes
        (E^{\lambda}_{p_\alpha\to\lambda,p_\beta\to\lambda})_{\mcO^n\mcF},
    \end{align}
    which is positive semidefinite.
    The branching trace identity~\eqref{eq:branching-trace-identity} implies
    \begin{align}
        \Tr_{\mcF}
        (E^\lambda_{p_\alpha\to\lambda,p_\beta\to\lambda})_{\mcO^n\mcF}
        =\delta_{\alpha\beta}
        \frac{\mD{\lambda}}{\mD{\alpha}}
        (E^\alpha_{p_\alpha p_\alpha})_{\mcO^n},
    \end{align}
    which gives the partial trace of $C_1$ as
    \begin{align}
        \Tr_{\mcF} C_1
        &=\frac{1}{M_{\lambda}}
        \sum_{\alpha\in\lambda-\square}
        \frac{d\mD{\alpha}}{\md{\lambda}} 
        (E^{\lambda}_{p_{\alpha}\to\lambda,p_{\alpha}\to\lambda})_{\mcI^n\mcP}
        \otimes(E^{\alpha}_{p_\alpha p_\alpha})_{\mcO^n}\\
        &\leq C \otimes \1_{\mcO^n},
        \label{eq:branching-partial-trace-inequality}
    \end{align}
    where $C$ is defined by
    \begin{align}
        C&\coloneqq \frac{1}{M_{\lambda}}
        \sum_{\alpha\in\lambda-\square}
        \frac{d\mD{\alpha}}{\md{\lambda}} 
        (E^{\lambda}_{p_\alpha\to\lambda,p_\alpha\to\lambda})_{\mcI^n\mcP},
    \end{align}
    and the inequality follows from $(E^\alpha_{p_\alpha p_\alpha})_{\mcO^n}\leq\1_{\mcO^n}$.
    We define an operator $J_{\sfS}$ by
    \begin{align}
        \label{eq:primal-solution}
        J_{\sfS}\coloneqq C_1+
        [C\otimes\1_{\mcO^n}-\Tr_{\mcF}C_1]
        \otimes\frac{\1_{\mcF}}{D},
    \end{align}
    which is positive semidefinite due to $C_1\geq 0$ and Eq.~\eqref{eq:branching-partial-trace-inequality}.
    Using Eq.~\eqref{eq:performance-operator-schur}, its objective value obeys
    \begin{align}
        \Tr(J_{\sfS}\Omega_{d,D,n})
        &\geq \Tr(C_1\Omega_{d,D,n})\\
        &={1\over d}
        \frac{\sum_{\alpha\in\lambda-\square}\mD{\alpha}}
        {\mD{\lambda}}\\
        &={1\over d}\max_{\lambda\vdash_d n+1}
        \frac{\sum_{\alpha\in\lambda-\square}\mD{\alpha}}
        {\mD{\lambda}}.
    \end{align}
    Finally, since
    \begin{align}
        [\Tr_{\mcI^n}
        (E^\lambda_{p_\alpha\to\lambda,p_\alpha\to\lambda})_{\mcI^n\mcP}, U_{\mcP}] = 0 \quad \forall U\in \U(d)
    \end{align}
    holds, we have
    \begin{align}
        \label{eq:branching-trace-identity-2}
        \Tr_{\mcI^n}
        (E^\lambda_{p_\alpha\to\lambda,p_\alpha\to\lambda})_{\mcI^n\mcP}
        =\frac{\md{\lambda}}{d}\1_{\mcP}.
    \end{align}
    The definition~\eqref{eq:primal-solution} of $J_{\sfS}$ implies the first constraint in Eq.~\eqref{eq:primal-sdp-parallel}:
    \begin{align}
        \Tr_{\mcF} J_{\sfS} &= C \otimes \1_{\mcO^n},
    \end{align}
    and Eq.~\eqref{eq:branching-trace-identity-2} implies the second constraint in Eq.~\eqref{eq:primal-sdp-parallel}:
    \begin{align}
        \Tr_{\mcI^n} C &= \1_{\mcP}.
    \end{align}
    Thus $J_{\sfS}$ is a feasible solution attaining the stated lower bound.
\end{proof}

\subsection{Optimality}

We construct a feasible solution of the dual SDP ~\eqref{eq:dual-sdp-general} to attain the matching upper bound in Eq.~\eqref{eq:max_form}.

\begin{proposition}
    \label{prop:optimality}
The optimal fidelity for complex conjugation of an unknown isometry $V:\CC^d\to\CC^D$ from $n$ calls with general protocols, including indefinite causal order, satisfies
\begin{align}
    F_{d,D,n}^{(\mathrm{GEN})}
    \leq {1\over d}\max_{\lambda\vdash_d n+1}
    \frac{\sum_{\alpha\in\lambda-\square}\mD{\alpha}}
    {\mD{\lambda}}.
\end{align}
\end{proposition}
\begin{proof}
    We consider the dual ansatz, analogous to that of Ref.~\cite{ebler2023optimal}, given by
    \begin{align}
        \label{eq:dual-ansatz-c}
        c&={1\over d}\max_{\lambda\vdash_d n+1}
        \frac{\sum_{\alpha\in\lambda-\square}\mD{\alpha}}
        {\mD{\lambda}},\\
        W&=\sum_{\substack{\alpha\vdash_d n\\
        p_\alpha,q_\alpha\in\Path{\alpha}}}
        \frac{(E^\alpha_{p_\alpha q_\alpha})_{\mcI^n}
        \otimes(E^\alpha_{p_\alpha q_\alpha})_{\mcO^n}
        \otimes\1_{\mcP}}
        {d\mD{\alpha}}.
    \end{align}
    For $\alpha\vdash_d n$ and $\lambda, \mu\in \alpha+\square$, we define positive semidefinite operators
    \begin{align}
        X_{\alpha;\lambda,\mu}
        &\coloneqq
        \sum_{p_\alpha, q_\alpha\in\Path{\alpha}}(E^\lambda_{p_\alpha\to\lambda,q_\alpha\to\lambda})_{\mcI^n\mcP} \otimes
        (E^\mu_{p_\alpha\to\mu,q_\alpha\to\mu})_{\mcO^n\mcF},\\
        Y_\lambda
        &\coloneqq \sum_{p_\lambda, q_\lambda\in\Path{\lambda}}
        (E^\lambda_{p_\lambda q_\lambda})_{\mcI^n\mcP}
        \otimes(E^\lambda_{p_\lambda q_\lambda})_{\mcO^n\mcF}.
        \label{eq:dual-block-abbreviation}
    \end{align}
    We first verify the operator inequality:
    \begin{align}
        W\otimes\1_{\mcF}
        &=\sum_{\substack{\alpha\vdash_d n,
        \lambda\in\alpha+_d\square\\
        \mu\in\alpha+_D\square}}
        \frac{X_{\alpha;\lambda,\mu}}
        {d\mD{\alpha}}\\
        &\geq\sum_{\lambda\vdash_d n+1,
        \alpha\in\lambda-\square}
        \frac{\mathsf X_{\alpha;\lambda,\lambda}}
        {d\mD{\alpha}}\\
        &\geq\sum_{\lambda\vdash_d n+1}
        \frac{Y_\lambda}
        {d\sum_{\alpha\in\lambda-\square}\mD{\alpha}}\\
        &\geq\frac{1}{cd^2}
        \sum_{\lambda\vdash_d n+1}
        \frac{Y_\lambda}{\mD{\lambda}}\\
        &=\frac{\Omega_{d,D,n}}{c}.
    \end{align}
    The first inequality discards the positive blocks with $\lambda\neq\mu$.
    For each fixed $\lambda$, the second inequality follows from Lem.~\ref{eq:matrix-inequality} applied to the vectors and coefficients given by
    \begin{align}
        \ket{a_\alpha^\lambda}
        &\coloneqq\frac{1}{\sqrt{\mD{\alpha}}}
        \sum_{p_\alpha\in\Path{\alpha}}
        \ket{p_\alpha\to\lambda}
        \otimes\ket{p_\alpha\to\lambda},\\
        c_\alpha^\lambda
        &\coloneqq\sqrt{\mD{\alpha}},\\
        \label{eq:dual-lemma-vectors}
        A&\coloneqq \{(\alpha,\lambda)\mid \alpha\in\lambda-\square\}.
    \end{align}
    The last inequality follows from the definition~\eqref{eq:dual-ansatz-c} of $c$, and Eq.~\eqref{eq:performance-operator-schur} gives the final equality.

    Since $g_\alpha(\sigma)$ is real orthogonal in the Young--Yamanouchi basis,
    \begin{align}
        [g_\alpha(\sigma)\otimes g_\alpha(\sigma)]
        \ket{\phi_\alpha^+}=\ket{\phi_\alpha^+}
        \quad\forall\sigma\in\mfS_n
    \end{align}
    holds, and hence
    \begin{align}
        [W,\pi_d(\sigma)_{\mcI^n}\otimes
        \pi_D(\sigma)_{\mcO^n}\otimes\1_{\mcP}]=0
        \quad\forall\sigma\in\mfS_n.
    \end{align}
    Thus, it is enough to check the second constraint in Eq.~\eqref{eq:dual-sdp-general} for the case of $i=n$:
    \begin{align}
        \Tr_{\mcO_n} W 
        &=\sum_{\substack{\gamma\vdash_d n-1,
        \alpha\in\gamma+_d\square\\
        p_\gamma,q_\gamma\in\Path{\gamma}}}
        \frac{(E^{\alpha}_{p_\gamma\to\alpha,q_\gamma\to\alpha})_{\mcI^n}}
        {d\mD{\gamma}}\notag \otimes(E^{\gamma}_{p_\gamma q_\gamma})_{\mcO^{n-1}}
        \otimes\1_{\mcP}\\
        &=\sum_{\substack{\gamma\vdash_d n-1\\
        p_\gamma,q_\gamma\in\Path{\gamma}}}
        \frac{(E^{\gamma}_{p_\gamma q_\gamma})_{\mcI^{n-1}}
        \otimes\1_{\mcI_n}}
        {d\mD{\gamma}} \otimes(E^{\gamma}_{p_\gamma q_\gamma})_{\mcO^{n-1}}
        \otimes\1_{\mcP}\\
        &=\Tr_{\mcI_n\mcO_n}W
        \otimes\frac{\1_{\mcI_n}}{d}.
    \end{align}
    Finally, the dimension counting in Eq.~\eqref{eq:schur-weyl} gives
    \begin{align}
        \Tr W
        = \sum_{\alpha\vdash_d n} \md{\alpha} d_\alpha = d^n,
    \end{align}
    where $d_\alpha$ is the dimension of the Specht module $\mcS_\alpha$ (see also Eq.~\eqref{eq:dimension-specht}).
    Thus $(c,W)$ is a feasible solution of the dual SDP attaining the claimed upper bound.
\end{proof}

\subsection{Efficient implementation of the optimal protocol}

We present a quantum circuit implementing the optimal isometry complex conjugation protocol (see also Fig.~\ref{fig:optimal-circuit}).

\begin{proposition}
    \label{prop:optimal-circuit}
    The quantum circuit in Fig.~\ref{fig:optimal-circuit} implements an optimal parallel isometry complex conjugation protocol with the fidelity in Eq.~\eqref{eq:max_form}.
\end{proposition}
\begin{proof}[Proof sketch]
    Equation~\eqref{eq:schur-isometry} shows that the $V^{\otimes n}$ transmits the symmetric-group register unchanged and applies $V_\alpha$ on each unitary-group register.
    We prepare the quantum state
    \begin{align}
        \ket{y_{\lambda^*}}
        \coloneqq
        \sum_{\alpha\in\lambda^*-\square}
        \sqrt{\frac{\mD{\alpha}}{M_{\lambda^*}}}\ket{\alpha},
        \label{eq:branch-control-state}
    \end{align}
    where $M_{\lambda^*}$ is defined in Eq.~\eqref{eq:definition_M}.
    Controlling the Schur blocks by the quantum state $\ket{y_{\lambda^*}}$ and sandwiching them between the two dual Clebsch--Gordan transforms provides a Stinespring dilation isometry for the output quantum channel:
    \begin{align}
        U_V
        &\coloneqq
        \sum_{\alpha\in\lambda^*-\square}
        \sqrt{\frac{\mD{\alpha}}{M_{\lambda^*}}}
        \left[\1_{\mcU_{\lambda^*}^{(d)}}\otimes
        \mathrm{CG}^{(D)\dagger}_{\lambda^*,\overline\square}
        (\ket{\alpha}\otimes V_\alpha)\right]
        \mathrm{CG}^{(d)\dagger}_{\lambda^*,\overline\alpha}.
        \label{eq:stinespring-isometry}
    \end{align}
    The output quantum channel is given by
    \begin{align}
        \Phi_V(\cdot)
        \coloneqq
        \Tr_{\mcU_{\lambda^*}^{(d)}\mcU_{\lambda^*}^{(D)}}
        [U_V(\mathord\cdot)U_V^\dagger].
    \end{align}
    As shown in Appendix~\ref{app:optimal-circuit}, the channel fidelity between $\Phi_V$ and the complex conjugate channel $\overline{\mcV}$ is given by
    \begin{align}
        F_{\mathrm{ch}}(\Phi_V,\overline{\mcV})
        ={1\over d}
        \frac{\sum_{\alpha\in\lambda^*-\square}\mD{\alpha}}
        {\mD{\lambda^*}},
    \end{align}
    for all isometry $V:\CC^d\to\CC^D$, which concludes the proof.
\end{proof}

\begin{corollary}
    \label{cor:query-complexity}
    For $0<\epsilon<1$, the complex conjugate of an unknown isometry $V:\CC^d \to \CC^D$ with diamond-distance error at most $\epsilon$ can be implemented with $n$ parallel queries satisfying
    \begin{align}
        n
        \leq
        \begin{cases}
            d-1 & (D=d)\\
            \frac{d(D-d)}{\epsilon} & (D>d)
        \end{cases}.
        \label{eq:query-complexity-isometry}
    \end{align}
    The circuit complexity is given by $O(\operatorname{poly}(D,1/\epsilon))$.
    Conversely, any general protocol for isometry complex conjugation with diamond-distance error at most $\epsilon$ requires at least
    \begin{align}
        n
        \geq
        \begin{cases}
            d-1-{\epsilon \over 1-\epsilon} & (D=d)\\
            d\qty[{(D-d)(1-\epsilon) \over \epsilon}-1] & (D>d)
        \end{cases}.
    \end{align}
\end{corollary}
\begin{proof}
    For $D=d$, Eq.~\eqref{eq:closed-form} gives exact conjugation at $n=d-1$.
    For $D>d$, we choose
    \begin{align}
        q_\epsilon
        &\coloneqq
        \max\left\{1,
        \left\lceil\frac{(D-d)(1-\epsilon)}{\epsilon}\right\rceil
        \right\},\\
        n&\coloneqq dq_\epsilon-1.
    \end{align}
    Then,  Eq.~\eqref{eq:closed-form} reduces to $F=q_\epsilon/(q_\epsilon+D-d)$ and gives $1-F\leq\epsilon$, and we have $n\leq d(D-d)/\epsilon$ by definition.
    Due to the unitary-group invariance of the optimal protocol shown in Fig.~\ref{fig:optimal-circuit}, the output quantum channel $\Phi_V$ satisfies
    \begin{align}
        \label{eq:unitary-symmetry}
        \Phi_{WVU} = \overline{\mcW} \circ \Phi_V \circ \overline{\mcU} \quad \forall U\in\U(d), W\in\U(D).
    \end{align}
    Then, we can show that the diamond-distance error is bounded by the channel fidelity~\cite[Lemma S2]{yoshida2026quantum}:
    \begin{align}
        \frac{1}{2}\|\Phi_V-\overline{\mcV}\|_\diamond
        &\leq 1-F_{\mathrm{ch}}(\Phi_V,\overline{\mcV}).
    \end{align}
    Thus, the diamond-distance error is at most $\epsilon$ with $n$ parallel queries.

    We next show that the quantum circuit in Fig.~\ref{fig:optimal-circuit} can be implemented efficiently for this choice of parameters.
    In this case, the balanced diagram is given by $\lambda^*=(q_\epsilon^d)$, and it has a single parent $\alpha^*=(q_\epsilon^{d-1},q_\epsilon-1)$, i.e., $\lambda^*-\square = \{\alpha^*\}$.
    Since $\mathrm{CG}_{\lambda^*, \overline{\alpha}^*}^{(d)} \cong \mathrm{CG}_{\lambda^*, \square}^{(d)}$ holds, the quantum circuit in Fig.~\ref{fig:optimal-circuit} can be implemented with efficient (dual) Clebsch--Gordan transforms~\cite{harrow2005applications,bacon2006efficient,bacon2007quantum,krovi2019efficient,burchardt2025high,nguyen2023mixed,grinko2023gelfand,fei2023efficient,grinko2025mixed}, which results in the circuit complexity $O(\operatorname{poly}(D,n))=O(\operatorname{poly}(D,1/\epsilon))$.

    We finally show a lower bound on the query complexity.
    Since ${1\over 2} \|\rho-\ketbra{\psi}\|_1 \geq 1-\bra{\psi}\rho\ket{\psi}$ holds for any quantum state $\rho$ and pure state $\ket{\psi}$, we have
    \begin{align}
        \frac{1}{2}\|\Phi_V-\overline{\mcV}\|_\diamond
        &\geq 1-F_{\mathrm{ch}}(\Phi_V,\overline{\mcV}).
    \end{align}
    Thus, the diamond-distance error is at most $\epsilon$ only if the channel fidelity is at least $1-\epsilon$.
    Then, Eq.~\eqref{eq:closed-form} gives the lower bound on the query complexity.
\end{proof}

\subsection{Application: complex conjugation of general quantum channels}

Let $V_\Lambda$ be a Stinespring dilation isometry of a channel $\Lambda$.
Then $\overline{V}_{\Lambda}$ is a Stinespring dilation isometry of the complex conjugate channel $\overline{\Lambda}$ defined by $\overline{\Lambda}(\cdot)\coloneqq \overline{\Lambda(\overline{\cdot})}$. Therefore, the complex conjugation of $\Lambda$ can be implemented by conjugating $V_\Lambda$ and subsequently tracing out the environment.
Since the complex conjugate of a quantum channel $\Lambda$ is a Stinespring dilation isometry of a complex conjugate channel $\overline{\Lambda}$ defined by $\overline{\Lambda}(\rho)\coloneqq \overline{\Lambda(\overline{\rho})}$, it can be implemented by the complex conjugate of its Stinespring isometry.
By combining isometry complex conjugation with the random dilation superchannel~\cite{girardi2025random,yoshida2025random}, we can implement the complex conjugate of a general quantum channel as follows.

\begin{corollary}
    \label{cor:channel-conjugation-achievability}
    Suppose we have an unknown quantum channel $\Lambda:\mcL(\CC^{d_\mathrm{in}})\to\mcL(\CC^{d_\mathrm{out}})$ with Kraus rank at most $r$.
    The complex conjugate of $\Lambda$ can be implemented with the diamond-distance error at most $\epsilon$ using at most
    \begin{align}
        n\leq
        \begin{cases}
            d_\mathrm{in}-1 & (\Delta_\Lambda=0)\\
            \frac{d_\mathrm{in}\Delta_\Lambda}{\epsilon} & (\Delta_\Lambda>0)
        \end{cases}
        = O\qty(d_\mathrm{in}\qty[{\Delta_\Lambda \over \epsilon}+1])
    \end{align}
    parallel queries, where $\Delta_\Lambda\coloneqq d_\mathrm{out}r-d_\mathrm{in} \geq 0$.
\end{corollary}
\begin{proof}
    The quantum channel $\Lambda$ has a Stinespring dilation isometry $V:\CC^{d_\mathrm{in}}\to\CC^{d_\mathrm{out}}\otimes\CC^r$ such that $\Lambda(\rho)=\Tr_{\CC^r}[V\rho V^\dagger]$.
    We define the set of all such dilations by
    \begin{align}
        \mathrm{Dil}_r(\Lambda)
        \coloneqq
        \bigl\{V\in \isometry{d_\mathrm{in}}{d_\mathrm{out}r}\mid
        \Tr_{\CC^r}[V(\mathord\cdot)V^\dagger]=\Lambda\bigr\},
    \end{align}
    which has a Haar measure induced by the Haar measure on $\U(r)$.
    The random dilation superchannel performs the exact conversion
    \begin{align}
        \Lambda^{\otimes n}
        \mapsto
        \mathbb E_{V\sim\mathrm{Dil}_r(\Lambda)}
        [\mcV^{\otimes n}],
    \end{align}
    where the expectation is taken over the Haar measure on $\mathrm{Dil}_r(\Lambda)$.
    Applying the optimal isometry complex conjugation protocol to the output of the random dilation superchannel gives a channel $\Phi_V$ for each $V\in\mathrm{Dil}_r(\Lambda)$, which satisfies
    \begin{align}
        \frac{1}{2} \|\Phi_V-\overline{\mcV}\|_\diamond\leq \epsilon.
    \end{align}
    Since $\Tr_{\CC^r}\circ\overline{\mcV}=\overline\Lambda$ holds, the data-processing inequality holds for the diamond norm; we have
    \begin{align}
        \frac{1}{2}\left\|
        \mathbb E_V[\Tr_{\CC^r}\circ\Phi_V]-\overline\Lambda
        \right\|_\diamond
        &\leq
        \mathbb E_V\left[
        \frac{1}{2}\|\Phi_V-\overline{\mcV}\|_\diamond
        \right]
        \leq\epsilon.
    \end{align}
\end{proof}

By embedding a suitably chosen isometry channel onto a quantum channel, we obtain the following lower bounds from the exact isometry fidelity in Thm.~\ref{thm:optimal-fidelity}.

\begin{corollary}
    \label{cor:channel-conjugation-lower-bound}
    Let $d_\mathrm{out}\geq2$, $r\geq1$, $1\leq d_\mathrm{in}\leq r d_\mathrm{out}$, and $0<\epsilon\leq1/4$, and define $\Delta_\Lambda\coloneqq r d_\mathrm{out}-d_\mathrm{in}$.
    Any protocol that implements complex conjugation with diamond-distance error at most $\epsilon$ for every channel $\Lambda:\mcL(\CC^{d_\mathrm{in}})\to\mcL(\CC^{d_\mathrm{out}})$ of Kraus rank at most $r$ requires
    \begin{align}
        n
        &\geq \Omega\!\left(
        d_\mathrm{out}\left[
        \frac{\min\{d_\mathrm{in},d_\mathrm{out},\Delta_\Lambda\}}
        {\epsilon}+1
        \right]
        \right).
        \label{eq:channel-conjugation-lower-bound-general-r}
    \end{align}
    queries.
    For a low rank quantum channel satisfying $r=O(1)$, this reduces to
    \begin{align}
        n
        \geq \Omega\!\left(d_\mathrm{in}
        \qty[{\Delta_\Lambda \over \epsilon}+1]        \right),
        \label{eq:channel-conjugation-lower-bound-constant-r}
    \end{align}
    which shows the tightness of the upper bound in Cor.~\ref{cor:channel-conjugation-achievability} up to a constant factor.
\end{corollary}
\begin{proof}
    First suppose that $d_\mathrm{in}<r d_\mathrm{out}$, and define the set of admissible input dimensions for the embedded isometry channel by
    \begin{align}
        \mcA_r \coloneqq
        \left\{a\in\ZZ \mid
        \max\{1,d_\mathrm{in}-(r-1)d_\mathrm{out}\}
        \leq a\leq
        \min\{d_\mathrm{in},d_\mathrm{out}-1\}
        \right\}.
        \label{eq:admissible-embedded-isometry-dimension}
    \end{align}
    Since $d_\mathrm{in}$, $d_\mathrm{out}$, and $r$ are integers and $d_\mathrm{in}<r d_\mathrm{out}$, the gap $\Delta_\Lambda$ satisfies $\Delta_\Lambda\geq1$.
    Thus, $1\leq\min\{d_\mathrm{in},d_\mathrm{out}-1\}$ and
    \begin{align}
        d_\mathrm{in}-(r-1)d_\mathrm{out}
        =d_\mathrm{out}-\Delta_\Lambda
        \leq\min\{d_\mathrm{in},d_\mathrm{out}-1\}
    \end{align}
    hold, which proves that the set $\mcA_r$ in Eq.~\eqref{eq:admissible-embedded-isometry-dimension} is not empty.
    Fix any $a\in\mcA_r$ and decompose the input space as $\CC^{d_\mathrm{in}}=\mcH_A\oplus\mcH_B$, where $\dim\mcH_A=a$ and $\dim\mcH_B=d_\mathrm{in}-a$.
    By definition of $\mcA_r$, $\dim\mcH_B\leq(r-1)d_\mathrm{out}$ holds, thus there exists a fixed quantum channel $\Gamma:\mcL(\mcH_B)\to\mcL(\CC^{d_\mathrm{out}})$ of Kraus rank at most $r-1$.
    For an unknown isometry operator $V:\mcH_A\to\CC^{d_\mathrm{out}}$, we define a quantum channel $\Lambda_V:\mcL(\CC^{d_\mathrm{in}})\to\mcL(\CC^{d_\mathrm{out}})$ by
    \begin{align}
        \Lambda_V(\cdot)
        \coloneqq
        V \Pi_A (\cdot) \Pi_A V^\dagger
        +\Gamma(\Pi_B (\cdot) \Pi_B),
        \label{eq:embedded-isometry-channel}
    \end{align}
    where $\Pi_A$ and $\Pi_B$ are the projectors onto $\mcH_A$ and $\mcH_B$.
    The quantum channel $\Lambda_V$ has Kraus rank at most $r$.
    One query to $\Lambda_V$ can be simulated with one query to $V$, while the restriction of $\overline{\Lambda_V}$ to $\mcH_A$ is exactly $\overline{\mcV}$.
    Hence an $n$-query channel-conjugation protocol with error at most $\epsilon$ would give an $n$-query protocol for conjugating every isometry $V:\CC^a\to\CC^{d_\mathrm{out}}$ with error $\leq \epsilon$ due to the data-processing inequality for the diamond norm.
    Thus, the lower bound on the query complexity for isometry complex conjugation in Cor.~\ref{cor:query-complexity} gives
    \begin{align}
        \label{eq:channel-conjugation-lower-bound-finite}
        n\geq a\left[
        \frac{(d_\mathrm{out}-a)(1-\epsilon)}{\epsilon}-1
        \right].
    \end{align}
    Since $d_\mathrm{out}-a\geq 1$ and $\epsilon\leq 1/4$ hold, we have
    \begin{align}
        n\geq {a(d_\mathrm{out}-a) \over 2\epsilon}.
    \end{align}
    By taking
    \begin{align}
        a =
        \begin{cases}
            d_\mathrm{in} & (d_\mathrm{in}<d_\mathrm{out}/2)\\
            d_\mathrm{out}-\Delta_\Lambda & (\Delta_\Lambda<d_\mathrm{out}/2)\\
            \lfloor d_\mathrm{out}/2\rfloor & (\text{otherwise})
        \end{cases},
    \end{align}
    we obtain
    \begin{align}
        n&\geq \Omega(d_\mathrm{out}\min\{d_\mathrm{in},d_\mathrm{out},\Delta_\Lambda\}/\epsilon)\\
        &\geq \Omega\qty(
        d_\mathrm{out}\qty[
        \frac{\min\{d_\mathrm{in},d_\mathrm{out},\Delta_\Lambda\}}
        {\epsilon}+1]).
    \end{align}

    We next consider the corner case $d_\mathrm{in}=r d_\mathrm{out}$.
    We identify the input space $\CC^{d_\mathrm{in}}$ with $\CC^{d_\mathrm{out}} \otimes \CC^r$, and for an unknown unitary $U\in\U(d_\mathrm{out})$, we define the rank-$r$ channel
    \begin{align}
        \Lambda_U(\cdot)
        \coloneqq
        U\Tr_{\CC^r}[\cdot]U^\dagger.
        \label{eq:embedded-unitary-channel-endpoint}
    \end{align}
    One query to $\Lambda_U$ can be simulated with one query to $U$, and restricting $\overline{\Lambda_U}$ to inputs $\rho\otimes\ketbra{1}$ gives $\overline U\rho U^\top$.
    Thus an $n$-query channel-conjugation protocol with error at most $\epsilon$ gives an $n$-query protocol for conjugating every unitary in dimension $d_\mathrm{out}$ with the error $\leq \epsilon$.
    Thus, the lower bound on the query complexity for unitary complex conjugation in Cor.~\ref{cor:query-complexity} gives
    \begin{align}
        n&\geq d_\mathrm{out}-1-{\epsilon \over 1-\epsilon}\\
        &\geq d_\mathrm{out}-1-{1\over 3}\\
        &\geq \Omega(d_\mathrm{out}).
    \end{align}
    Since $\Delta_\Lambda=0$ holds in this case, we obtain Eq.~\eqref{eq:channel-conjugation-lower-bound-general-r}.

    Finally, for the case of $r=O(1)$, since $\Delta_\Lambda = rd_\mathrm{out}-d_\mathrm{in}\geq 0$, we have $d_\mathrm{out} = \Omega(d_\mathrm{in})$ and $\Delta_\Lambda = O(d_\mathrm{out})$.
    Thus, we have
    \begin{align}
        d_\mathrm{out}\qty[\frac{\min\{d_\mathrm{in},d_\mathrm{out},\Delta_\Lambda\}}{\epsilon}+1]
        &=
        \begin{cases}
            \Omega\qty(d_\mathrm{in} \qty[{\Delta_\Lambda \over \epsilon}+1]) & (\Delta_\Lambda = O(d_\mathrm{in}))\\
            \Omega\qty(d_\mathrm{in} \qty[{d_\mathrm{out} \over \epsilon}+1]) & (\Delta_\Lambda = \Omega(d_\mathrm{in}))
        \end{cases}\\
        &= \Omega\qty(d_\mathrm{in} \qty[{\Delta_\Lambda \over \epsilon}+1]),
    \end{align}
    which gives Eq.~\eqref{eq:channel-conjugation-lower-bound-constant-r}.
\end{proof}

\section{Extension to multi-copy isometry complex conjugation}
\label{sec:multi-copy-conjugation}

We consider the extension of isometry complex conjugation to the multi-copy case, which is to implement $f_k: V\mapsto \overline{V}^{\otimes k}$ for $k\geq 1$ using $n$ queries of the unknown isometry $V: \CC^d \to \CC^D$.
We define the average fidelity for the multi-copy case by
\begin{align}
    F_{d,D,n,k}^{(X)}
    \coloneqq
    \sup_{\mathsf S\in X}
    F_\mathrm{avg}(\sfS, f_k)
    \label{eq:multi-copy-fidelity-definition}
\end{align}
for $X\in \{\mathrm{PAR}, \mathrm{GEN}\}$.
Then, we can show that the parallel estimation-based protocol achieves the asymptotically optimal fidelity:

\begin{theorem}
\label{thm:multi-copy-conjugation-asymptotic}
For fixed $d,D,k$ with $d<D$, the asymptotic multi-copy conjugation fidelity is given by
\begin{align}
    \label{eq:multi-copy-conjugation-asymptotic}
    F_{d,D,n,k}^{(\mathrm{GEN})}
    =1-\frac{kd(D-d)}{n}+o(n^{-1}),
\end{align}
which is asymptotically achieved by a parallel estimation-based protocol.
\end{theorem}
\begin{proof}[Proof sketch]
    As shown in Ref.~\cite{yoshida2026quantum}, the optimal estimation of an unknown isometry $V:\CC^d\to\CC^D$ from $n$ queries achieves the leading-order fidelity
    \begin{align}
        F_\mathrm{est}(n,d,D)
        =1-\frac{d(D-d)}{n}+o(n^{-1}),
        \label{eq:estimation-leading-order}
    \end{align}
    where $F_\mathrm{est}$ is the average fidelity defined by
    \begin{align}
        F_\mathrm{est}(n,d,D) \coloneqq \int_{\isometry{d}{D}} \dd V \sum_i p(\hat{V}_i|V) F_\mathrm{ch}(\hat{V}_i, V),
    \end{align}
    where $p(\hat{V}_i|V)$ is the probability of obtain an estimator $\hat{V}_i$ given the true isometry $V$.
    By using this estimation strategy, we can construct a parallel protocol for the multi-copy conjugation task, which is to estimate the unknown isometry $V$ and then prepare the conjugate of the estimated isometry $\hat{V}$ for $k$ copies.
    In this case, the average fidelity of the multi-copy conjugation task can be expressed as
    \begin{align}
        F_{d,D,n,k}^{(\mathrm{PAR})}
        &\geq \int_{\isometry{d}{D}} \dd V \sum_i p(\hat{V}_i|V) F_\mathrm{ch}(\overline{\hat{V}}_i^{\otimes k}, \overline{V}^{\otimes k})\\
        &=\int_{\isometry{d}{D}} \dd V \sum_i p(\hat{V}_i|V) F_\mathrm{ch}(\hat{V}_i, V)^k.
    \end{align}
    By using Bernoulli's inequality $X^k\geq 1-k(1-X)$ for $X\in [0,1]$, we have
    \begin{align}
        F_{d,D,n,k}^{(\mathrm{PAR})}&\geq \int_{\isometry{d}{D}} \dd V \sum_i p(\hat{V}_i|V) [1-k(1-F_\mathrm{ch}(\hat{V}_i, V))]\\
        &=1-k[1-F_\mathrm{est}(n,d,D)]\\
        &=1-\frac{kd(D-d)}{n}+o(n^{-1}).
        \label{eq:multi-copy-conjugation-asymptotic-lower-bound}
    \end{align}
    Appendix~\ref{app:multi-copy-optimality} proves the matching upper bound by constructing a feasible solution of the dual SDP for the multi-copy conjugation task.
\end{proof}

\section{Discussion}
\label{sec:discussion}

\subsection{Comparison with estimation-based protocol}

\begin{figure}
    \centering
    \includegraphics[width=.5\linewidth]{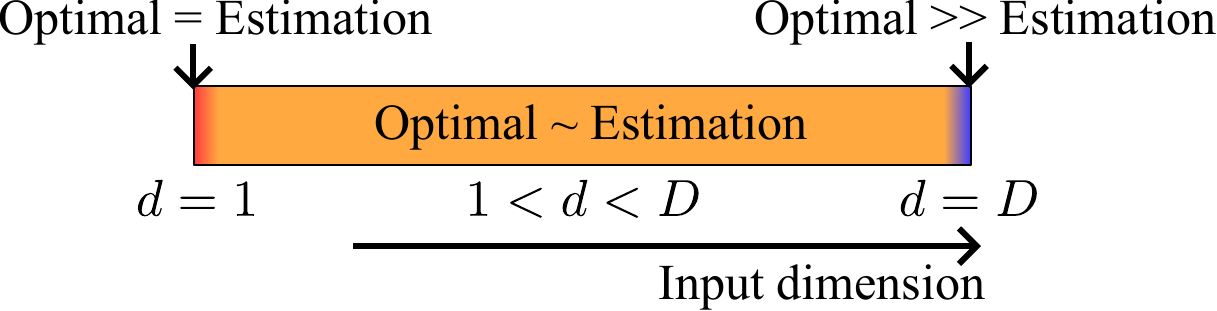}
    \caption{Phase diagram of the optimality of estimation-based strategies for state transposition, isometry complex conjugation, and unitary complex conjugation. The estimation-based strategy is optimal for state transposition~\cite{buzek1999optimal,brzic2026optimal} and asymptotically optimal for isometry complex conjugation (this work), while it is suboptimal even asymptotically for unitary complex conjugation~\cite{miyazaki2019complex,ebler2023optimal}.
    This figure is modified from a similar diagram in Ref.~\cite{yoshida2026quantum}.}
    \label{fig:phase_diagram}
\end{figure}

\begin{figure}
    \centering
    \includegraphics[width=.5\linewidth]{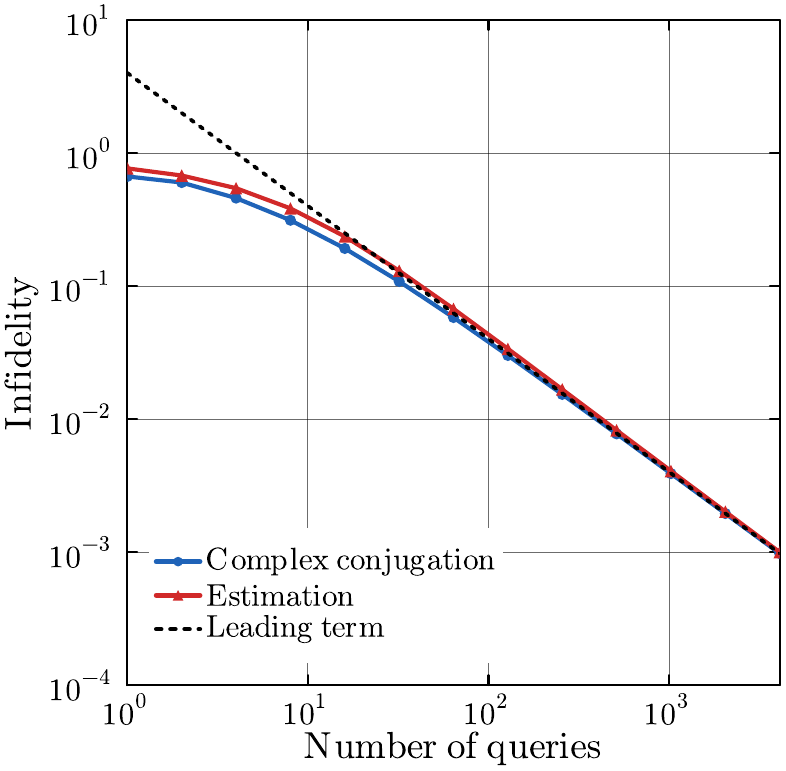}
    \caption{We compare the optimal infidelity of isometry complex conjugation (blue circles) and the isometry estimation (red triangles) for the case of $d=2$, $D=4$.
    The optimal infidelity of isometry complex conjugation is lower than that of the isometry estimation for finite $n$, which shows that the optimal isometry complex conjugation protocol has an advantage over the estimation-based strategy for finite $n$.
    Both infidelities approach ${d(D-d)/n} = 4/n$ (dotted line) in the limit of large $n$.}
    \label{fig:comparison}
\end{figure}

In Sec.~\ref{sec:optimal-complex-conjugation}, we provide a quantum circuit implementing the optimal isometry complex conjugation protocol (Fig.~\ref{fig:optimal-circuit}), which is not an estimation-based strategy.
On the other hand, Theorem~\ref{thm:multi-copy-conjugation-asymptotic} shows that the estimation-based strategy is \emph{asymptotically} optimal for multi-copy isometry complex conjugation for the case of $D>d$.
This result connects the state transposition and unitary complex conjugation, where the estimation-based strategy is known to be optimal for the state transposition~\cite{buzek1999optimal,brzic2026optimal}.
At the same time, it is suboptimal even asymptotically for the unitary complex conjugation~\cite{miyazaki2019complex,ebler2023optimal} (see Fig.~\ref{fig:phase_diagram}).

Here, we discuss the advantage of this circuit over an estimation-based strategy.
The estimation-based strategy can be strictly suboptimal at finite $n$.
We numerically compare the optimal isometry complex conjugation protocol and the estimation-based strategy for $d=2$, $D=4$, and one target copy in Fig.~\ref{fig:comparison}, which shows that the optimal isometry complex conjugation protocol has a better fidelity than the estimation-based strategy for finite $n$.
To obtain the optimal fidelity of isometry estimation, we find the maximum eigenvalue of the estimation matrix $M_\mathrm{est}(n,d,D)$ defined by~\cite{yoshida2026quantum}
\begin{align}
    (M_\mathrm{est}(n,d,D))_{\alpha\beta}
    \coloneqq
    {1\over d^2}\sum_{i,j=1}^{d} \delta_{\alpha+e_i, \beta+e_j} \sqrt{(\alpha_i-i+d+1)(\beta_j-j+d+1)\over (\alpha_i-i+D+1)(\beta_j-j+D+1)} \quad \forall \alpha, \beta\vdash_d n.
\end{align}
For the case of $d=2$, the estimation matrix $M_\mathrm{est}(n,2,D)$ is a $(\lfloor n/2\rfloor +1) \times (\lfloor n/2\rfloor +1)$ tridiagonal matrix. Its maximum eigenvalue can be computed in $O(n \log 1/\epsilon)$ time with precision $\epsilon$ using the Sturm bisection method~\cite{kahan1966accurate}.
We use the function \texttt{LinearAlgebra.eigmax} in Julia, with the matrix represented as a \texttt{SymTridiagonal} object.
The source code for the numerical calculation is available at Ref.~\cite{github}, which can be freely used under the MIT license~\cite{mit}.

In addition, the optimal isometry complex conjugation protocol has a polynomial-size implementation, while the optimal covariant estimator uses a continuous-outcome measurement whose efficient realization is not known~\cite{yoshida2026quantum}.
Thus, the optimal isometry complex conjugation protocol has an advantage over the estimation-based strategy for finite $n$ in terms of both fidelity and circuit complexity.

\subsection{Comparison of channel complex conjugation protocols}

We construct the channel complex conjugation by combining the optimal isometry complex conjugation protocol with the random dilation superchannel~\cite{girardi2025random,yoshida2025random}.
By using the isometry-to-unitary lifting~\cite{yoshida2025random}, we can also construct the channel complex conjugation protocol by combining the optimal unitary complex conjugation protocol~\cite{miyazaki2019complex,ebler2023optimal} with the random dilation superchannel.
Since $m$ queries of the random dilation unitary $U$ is obtained from $O(m^2 d_\mathrm{in}^2/\epsilon)$ queries of the original quantum channel $\Lambda$, the unitary-based channel complex conjugation protocol requires $O(d_\mathrm{in}^2 d_\mathrm{out}^2 r^2/\epsilon)$ queries of the original quantum channel $\Lambda$ to achieve the diamond-distance error $\epsilon$.
On the other hand, the isometry-based channel complex conjugation protocol requires $O(d_\mathrm{in} d_\mathrm{out} r/\epsilon)$ queries of the original quantum channel $\Lambda$ (see Cor.~\ref{cor:channel-conjugation-achievability}).
Thus, our isometry-based channel complex conjugation protocol has a quadratic improvement over the query bound obtained via the dilation-unitary lifting.

\section{Conclusion}
\label{sec:conclusion}

We obtain the exact optimal fidelity for conjugating an unknown isometry from $n$ queries and prove that a parallel protocol achieves the optimal fidelity among all general protocols, including those with indefinite causal order.
The closed formula recovers the optimal unitary complex conjugation~\cite{miyazaki2019complex,ebler2023optimal} and optimal state transposition~\cite{buzek1999optimal,brzic2026optimal}.
The optimal channel can be implemented using Schur and (dual) Clebsch--Gordan transforms, which can be realized with the circuit complexity $O(\operatorname{poly}(D,1/\epsilon))$ for diamond-distance error $\epsilon$.
Combining this circuit with a random dilation superchannel gives a complex conjugation protocol of an unknown quantum channel, which is optimal for the case of constant Kraus rank $r=O(1)$ up to a constant factor.
We extend the result to multi-copy complex conjugation $V^{\otimes n}\mapsto\overline V^{\otimes k}$, and show that the estimation-based strategy is asymptotically optimal for $D>d$.

We conclude with three open questions for future work.
First, determining the exact multi-copy optimum would extend both the one-copy result here and optimal multi-copy state transposition~\cite{brzic2026optimal}.
Second, the complementary superchannel of the optimal conjugation protocol may clarify its relation to isometry/channel cloning~\cite{sekatski2025cloning}, paralleling the complementarity between optimal state transposition and cloning~\cite{brzic2026optimal,chiribella2008optimal}.
Finally, our lower bound on the query complexity of channel complex conjugation motivates showing a separation between two access models: one can only query $\Lambda$, and the other can query $\Lambda$ and $\overline{\Lambda}$, for a certain learning task of quantum channels.

\section*{Note added}

During the preparation of this manuscript, we became aware of an independent work~\cite{adhikari2026prep}, which derives upper and lower bounds of the optimal fidelity for multi-copy unitary complex conjugation.
This is a special case of our work for $D=d$.
The upper bound in their work is equivalent to that in our work (Prop.~\ref{prop:multi-copy-optimality}) for the case of $D=d$.

\appendices

\section{Proof of Prop.~\ref{prop:optimal-circuit}}
\label{app:optimal-circuit}
\begin{proof}
    We define a linear operator $T_{\alpha, V}: \CC^d \to \mcU_{\lambda^*}^{(d)} \otimes \CC^D \otimes \mcU_{\lambda^*}^{(D)}$ for $\alpha\in \lambda^*-\square$ by
    \begin{align}
        T_{\alpha,V}
        \coloneqq
        \left[\1_{\mcU_{\lambda^*}^{(d)}}\otimes
        \mathrm{CG}^{(D)\dagger}_{\lambda^*,\overline\square}
        (\ket{\alpha}\otimes V_\alpha)\right]
        \mathrm{CG}^{(d)\dagger}_{\lambda^*,\overline\alpha}.
    \end{align}
    The isometry in Eq.~\eqref{eq:stinespring-isometry} is given by $U_V=M_{\lambda^*}^{-1/2}\sum_\alpha\sqrt{\mD{\alpha}}T_{\alpha,V}$.
    Then, the Choi operator of the output channel $\Phi_V$ is given by
    \begin{align}
        J_{\Phi_V}
        =\frac{1}{M_{\lambda^*}}
        \sum_{\alpha,\beta\in\lambda^*-\square}
        \sqrt{\mD{\alpha}\mD{\beta}}\,
        J_{\alpha\beta}(V),
        \label{eq:choi-branch-sum}
    \end{align}
    where $J_{\alpha\beta}(V)$ is defined by
    \begin{align}
        J_{\alpha\beta}(V)
        \coloneqq
        \Tr_{\mcU_{\lambda^*}^{(d)}\mcU_{\lambda^*}^{(D)}}
        [\dketbra{T_{\alpha,V}}{T_{\beta,V}}].
        \label{eq:branch-choi-block}
    \end{align}
    Then, the channel fidelity between $\Phi_V$ and $\overline{\mcV}$ is given by
    \begin{align}
        F_{\mathrm{ch}}(\Phi_V,\overline{\mcV})
        &= \frac{1}{d^2} \dbra{\overline V}J_{\Phi_V}\dket{\overline V}\\
        &= \frac{1}{d^2 M_{\lambda^*}} \sum_{\alpha, \beta\in\lambda^*-\square} \sqrt{\mD{\alpha}\mD{\beta}} \dbra{\overline V}J_{\alpha\beta}(V)\dket{\overline V}.
    \end{align}
    To calculate the inner product $\dbra{\overline V}J_{\alpha\beta}(V)\dket{\overline V}$, we use the following two identities:
    \begin{align}
        (\mathrm{CG}_{\lambda^*, \overline{\alpha}}^{(d)\dagger}\otimes \1_d)\dket{\1_d}
        &= \sum_{i=1}^{d} \sum_{\psi_{\lambda^*}=1}^{\md{\lambda^*}} \sum_{u_\alpha=1}^{\md{\alpha}} (\ketbra{u_{\lambda^*}} \otimes \ketbra{u_\alpha})\mathrm{CG}_{\lambda^*, \overline{\alpha}}^{(d)\dagger} \ket{i} \otimes \ket{i}\\
        &= \sqrt{\frac{d}{\md{\lambda^*}}} \sum_{i=1}^{d} \sum_{\psi_{\lambda^*}=1}^{\md{\lambda^*}} \sum_{u_\alpha=1}^{\md{\alpha}} \ketbra{u_{\lambda^*}} \mathrm{CG}_{\alpha, \square}^{(d)} (\ket{u_\alpha}\otimes \ket{i}) \otimes \ket{u_\alpha} \otimes \ket{i}\\
        &= \sqrt{\frac{d}{\md{\lambda^*}}} (\Pi_{\lambda^*}^{(d)}\mathrm{CG}_{\alpha, \square}^{(d)} \otimes \1_{\mcU_\alpha^{(d)}}\otimes \1_d) (\dket{\1_{\mcU_\alpha^{(d)}}} \otimes \dket{\1_d}),\label{eq:cg-identity-1}\\
        (\1_{\mcU_{\lambda^*}^{(D)}} \otimes \dbra{\1_D})(\mathrm{CG}_{\lambda, \overline{\square}}^{(D) \dagger} \otimes \1_D)
        &= \sum_{i=1}^{D} \sum_{\psi_{\lambda^*}=1}^{\mD{\lambda^*}} \sum_{u_\alpha=1}^{\mD{\alpha}} (\ketbra{u_\lambda^*} \otimes \bra{i}) \mathrm{CG}_{\lambda^*, \overline{\square}}^{(D) \dagger} \ketbra{u_\alpha} \otimes \bra{i}\\
        &= \sqrt{\frac{\mD{\alpha}}{\mD{\lambda^*}}} \sum_{i=1}^{D} \sum_{\psi_{\lambda^*}=1}^{\mD{\lambda^*}} \sum_{u_\alpha=1}^{\mD{\alpha}} \ketbra{u_{\lambda^*}} \mathrm{CG}_{\alpha,\square} (\ketbra{u_\alpha} \otimes \ketbra{i})\\
        &= \sqrt{\frac{\mD{\alpha}}{\mD{\lambda^*}}} \Pi_{\lambda^*}^{(D)}\mathrm{CG}_{\alpha, \square}^{(D)},
        \label{eq:cg-identity-2}
    \end{align}
    which follow from the property of the dual CG transform given in Eq.~\eqref{eq:dual-CG}, and $\Pi_{\lambda^*}^{(d)}$ and $\Pi_{\lambda^*}^{(D)}$ are projectors onto the irrep space $\mcU_{\lambda^*}^{(d)}$ and $\mcU_{\lambda^*}^{(D)}$, respectively.
    Then,
    \begin{align}
        (\1_{\mcU_{\lambda^*}^{(d)} \mcU_{\lambda^*}^{(D)}}\otimes \dbra{\overline{V}})\dket{T_{\alpha, V}}
        &= \left(\1_{\mcU_{\lambda^*}^{(d)} \mcU_{\lambda^*}^{(D)}}\otimes \dbra{\1_D}\right)\left[\1_{\mcU_{\lambda^*}^{(d)}}\otimes
        \mathrm{CG}^{(D)\dagger}_{\lambda^*,\overline\square}
        (\ket{\alpha}\otimes V_\alpha)\right]
        (\mathrm{CG}^{(d)\dagger}_{\lambda^*,\overline\alpha}\otimes V)\dket{\1_d}\\
        &= \sqrt{\frac{d \mD{\alpha}}{\md{\lambda^*}\mD{\lambda^*}}} (\Pi_{\lambda^*}^{(d)}\mathrm{CG}_{\alpha, \square}^{(d)} \otimes \Pi_{\lambda^*}^{(D)} \mathrm{CG}_{\alpha, \square}^{(D)}) (\dket{V_\alpha} \otimes \dket{V})\\
        &= \sqrt{\frac{d \mD{\alpha}}{\md{\lambda^*}\mD{\lambda^*}}} \dket{(\Pi_{\lambda^*}^{(D)} \mathrm{CG}_{\alpha, \square}^{(D)}) V_\alpha \otimes V (\mathrm{CG}_{\alpha, \square}^{(d)\dagger}\Pi_{\lambda^*}^{(d)})}\\
        &= \sqrt{\frac{d \mD{\alpha}}{\md{\lambda^*}\mD{\lambda^*}}} \dket{V_{\lambda^*}}
    \end{align}
    holds, where the second equality follows from Eqs.~\eqref{eq:cg-identity-1} and \eqref{eq:cg-identity-2}, and the last equality follows from Eq.~\eqref{eq:schur-isometry}.
    Then, we have
    \begin{align}
        \dbra{\overline V}J_{\alpha\beta}(V)\dket{\overline V}
        &= \Tr_{\mcU_{\lambda^*}^{(d)}\mcU_{\lambda^*}^{(D)}}[(\1_{\mcU_{\lambda^*}^{(d)}\mcU_{\lambda^*}^{(D)}}\otimes \dbra{\overline V})\dketbra{T_{\alpha,V}}{T_{\beta,V}}(\1_{\mcU_{\lambda^*}^{(d)}\mcU_{\lambda^*}^{(D)}}\otimes \dket{\overline V})]\\
        &= \frac{d \sqrt{\mD{\alpha}\mD{\beta}}}{\md{\lambda^*}\mD{\lambda^*}} \dbraket{V_{\lambda^*}}\\
        &= \frac{d \sqrt{\mD{\alpha}\mD{\beta}}}{\md{\lambda^*}\mD{\lambda^*}} \md{\lambda^*}.
    \end{align}
    Thus, we obtain
    \begin{align}
        F_{\mathrm{ch}}(\Phi_V,\overline{\mcV})
        &= \frac{1}{d^2 M_{\lambda^*}} \sum_{\alpha, \beta\in\lambda^*-\square} \sqrt{\mD{\alpha}\mD{\beta}} \frac{d \sqrt{\mD{\alpha}\mD{\beta}}}{\md{\lambda^*}\mD{\lambda^*}} \md{\lambda^*}\\
        &= \frac{1}{d M_{\lambda^*} \mD{\lambda^*}} \sum_{\alpha, \beta\in\lambda^*-\square} \mD{\alpha}\mD{\beta}\\
        &= \frac{M_{\lambda^*}^2}{d M_{\lambda^*} \mD{\lambda^*}}\\
        &= {1\over d} \frac{\sum_{\alpha\in\lambda^*-\square} \mD{\alpha}}{\mD{\lambda^*}},
    \end{align}
    which shows that the quantum circuit in Fig.~\ref{fig:optimal-circuit} implements an optimal parallel isometry complex conjugation protocol with the fidelity in Eq.~\eqref{eq:max_form}.
\end{proof}

\section{Proof of optimality in Thm.~\ref{thm:multi-copy-conjugation-asymptotic}}
\label{app:multi-copy-optimality}

For a multi-copy isometry complex conjugation $f_k: V\mapsto \overline{V}^{\otimes k}$ for $V\in \isometry{d}{D}$, the performance operator in Eq.~\eqref{eq:performance-operator} becomes
\begin{align}
\Omega_{d,D,n,k}
&\coloneqq \Omega_{f_k, n}\\
&=
{1 \over d^{2k}} \sum_{\substack{\lambda\vdash_d n+k\\p,q\in\Path{\lambda}}} {(E^\lambda_{pq})_{\mcI^n\mcP^k} \otimes(E^\lambda_{pq})_{\mcO^n\mcF^k} \over \mD{\lambda}}.
\label{eq:multi-copy-performance-operator}
\end{align}
We generalize Proposition~\ref{prop:optimality} by constructing a feasible solution of the dual SDP for the multi-copy conjugation task, which gives an upper bound of the optimal fidelity.

\begin{proposition}
\label{prop:multi-copy-optimality}
The optimal fidelity of multi-copy isometry complex conjugation in the general protocol is upper bounded by
\begin{align}
    \label{eq:multi-copy-optimality}
    F_{d,D,n,k}^{(\mathrm{GEN})}
    \leq
    \max_{\lambda\vdash_d n+k}
    \sum_{\alpha\vdash_d n}
    \frac{f^{\lambda/\alpha}\mD{\alpha}}
    {d^k\mD{\lambda}},
\end{align}
where $f^{\lambda/\alpha}\coloneqq |\Path{\alpha\to\lambda}|$ counts the $k$-step branching paths in the Young lattice.
\end{proposition}
\begin{proof}
    Similarly to the proof of Proposition~\ref{prop:optimality}, we construct a feasible solution of the dual SDP for the multi-copy conjugation task.
    The Hilbert spaces $\mcP$ and $\mcF$ in the dual SDP~\eqref{eq:dual-sdp-general} are given by $\mcP = \mcP^k \coloneqq \bigotimes_{i=1}^{k} \mcP_i$ and $\mcF = \mcF^k \coloneqq \bigotimes_{i=1}^{k} \mcF_i$ using $\mcP_i\cong \CC^d$ and $\mcF_i \cong \CC^D$ for $i\in [k]$.
    We construct a feasible solution of the dual SDP by
    \begin{align}
        c &= \max_{\lambda\vdash_d n+k}
        \sum_{\alpha\vdash_d n}
        \frac{f^{\lambda/\alpha}\mD{\alpha}}
        {d^k\mD{\lambda}},\\
        W&=\sum_{\substack{\alpha\vdash_d n\\
        p_\alpha,q_\alpha\in\Path{\alpha}}}
        \frac{(E^\alpha_{p_\alpha q_\alpha})_{\mcI^n}
        \otimes(E^\alpha_{p_\alpha q_\alpha})_{\mcO^n}
        \otimes\1_{\mcP^k}}
        {d^k\mD{\alpha}}.
    \end{align}
    We first verify the operator inequality:
    \begin{align}
        W\otimes \1_{\mcF^k}
        &=\sum_{\substack{\alpha\vdash_d n\\
        p_\alpha,q_\alpha\in\Path{\alpha}}}
        \frac{(E^\alpha_{p_\alpha q_\alpha})_{\mcI^n}
        \otimes\1_{\mcP^k}
        \otimes(E^\alpha_{p_\alpha q_\alpha})_{\mcO^n}
        \otimes\1_{\mcF^k}}
        {d^k\mD{\alpha}}\\
        &= \sum_{\substack{\alpha\vdash_d n\\p_\alpha,q_\alpha\in\Path{\alpha}}}
        \sum_{\substack{\lambda\vdash_d n+k,\ \mu\vdash_D n+k\\p_{\alpha\to\lambda}\in\Path{\alpha\to\lambda}\\p_{\alpha\to\mu}\in\Path{\alpha\to\mu}}}
        \frac{(E^\lambda_{p_\alpha\cup p_{\alpha\to\lambda},q_\alpha\cup p_{\alpha\to\lambda}})_{\mcI^n\mcP^k}
        \otimes(E^\mu_{p_\alpha\cup p_{\alpha\to\mu},q_\alpha\cup p_{\alpha\to\mu}})_{\mcO^n\mcF^k}}
        {d^k\mD{\alpha}}\\
        &\geq \sum_{\substack{\alpha\vdash_d n\\p_\alpha,q_\alpha\in\Path{\alpha}}}
        \sum_{\substack{\lambda\vdash_d n+k\\p_{\alpha\to\lambda}\in\Path{\alpha\to\lambda}}}
        \frac{(E^\lambda_{p_\alpha\cup p_{\alpha\to\lambda},q_\alpha\cup p_{\alpha\to\lambda}})_{\mcI^n\mcP^k}
        \otimes(E^\lambda_{p_\alpha\cup p_{\alpha\to\lambda},q_\alpha\cup p_{\alpha\to\lambda}})_{\mcO^n\mcF^k}}
        {d^k\mD{\alpha}}\\
        &\geq \sum_{\substack{\lambda\vdash_d n+k\\p_\lambda,q_\lambda\in\Path{\lambda}}}
        \frac{(E^\lambda_{p_\lambda q_\lambda})_{\mcI^n\mcP^k}
        \otimes(E^\lambda_{p_\lambda q_\lambda})_{\mcO^n\mcF^k}}
        {d^k\sum_{\alpha\vdash_d n}
        f^{\lambda/\alpha}\mD{\alpha}}\\
        &\geq\frac{1}{cd^{2k}}
        \sum_{\substack{\lambda\vdash_d n+k\\
        p_\lambda,q_\lambda\in\Path{\lambda}}}
        \frac{(E^\lambda_{p_\lambda q_\lambda})_{\mcI^n\mcP^k}
        \otimes(E^\lambda_{p_\lambda q_\lambda})_{\mcO^n\mcF^k}}
        {\mD{\lambda}}\\
        &=\frac{1}{c}\Omega_{d,D,n,k}.
    \end{align}
    The second equality follows from Thm.~\ref{thm:branching-tensor-identity}, the first inequality keeps only the case of $\lambda=\mu$ and $p_{\alpha\to\lambda} = p_{\alpha\to\mu}$, the second inequality uses Lem.~\ref{eq:matrix-inequality} applied to the set of vectors and coefficients given by
    \begin{align}
        \ket{a_{\alpha,p_{\alpha\to\lambda}}^\lambda}
        \coloneqq
        \frac{1}{\sqrt{\mD{\alpha}}}
        \sum_{p_\alpha\in\Path{\alpha}}
        \ket{p_\alpha\cup p_{\alpha\to\lambda}}\otimes\ket{p_\alpha\cup p_{\alpha\to\lambda}},
        \qquad
        c_{\alpha,p_{\alpha\to\lambda}}^\lambda\coloneqq\sqrt{\mD{\alpha}},\\
        A\coloneqq \{(\alpha, p_{\alpha\to\lambda})\mid \alpha\vdash_d n,p_{\alpha\to\lambda}\in\Path{\alpha\to\lambda}\},
    \end{align}
    and the third inequality uses the definition of $c$.
    As in Proposition~\ref{prop:optimality}, since
    \begin{align}
        [W,\pi_d(\sigma)_{\mcI^n}\otimes
        \pi_D(\sigma)_{\mcO^n}\otimes\1_{\mcP^k}]=0
        \quad\forall\sigma\in\mfS_n
    \end{align}
    holds, it is sufficient to check the second constraint in Eq.~\eqref{eq:dual-sdp-general} for the case of $i=n$:
    \begin{align}
        \Tr_{\mcO_n} W
        &=\sum_{\substack{\gamma\vdash_d n-1,
        \alpha\in\gamma+_d\square\\
        p_\gamma,q_\gamma\in\Path{\gamma}}}
        \frac{(E^{\alpha}_{p_\gamma\to\alpha,q_\gamma\to\alpha})_{\mcI^n}
        \otimes(E^{\gamma}_{p_\gamma q_\gamma})_{\mcO^{n-1}}
        \otimes\1_{\mcP^k}}
        {d^k\mD{\gamma}}\\
        &=\sum_{\substack{\gamma\vdash_d n-1\\
        p_\gamma,q_\gamma\in\Path{\gamma}}}
        \frac{(E^{\gamma}_{p_\gamma q_\gamma})_{\mcI^{n-1}}
        \otimes\1_{\mcI_n}
        \otimes(E^{\gamma}_{p_\gamma q_\gamma})_{\mcO^{n-1}}
        \otimes\1_{\mcP^k}}
        {d^k\mD{\gamma}}\\
        &=\Tr_{\mcI_n\mcO_n}W
        \otimes\frac{\1_{\mcI_n}}{d}.
    \end{align}
    Finally, the dimension counting in Eq.~\eqref{eq:schur-weyl} gives
    \begin{align}
        \Tr W
        &= \sum_{\alpha\vdash_d n} \md{\alpha} d_\alpha = d^n.
    \end{align}
    Thus $(c,W)$ is feasible for the dual SDP and proves Eq.~\eqref{eq:multi-copy-optimality}.
\end{proof}

\begin{proposition}
    \label{prop:multi-copy-optimality-asymptotic}
For fixed $d,D,k$ and $n\to\infty$, the right-hand side of Eq.~\eqref{eq:multi-copy-optimality} satisfies
\begin{align}
    \max_{\lambda\vdash_d n+k}
    \sum_{\alpha\vdash_d n}
    \frac{f^{\lambda/\alpha}\mD{\alpha}}
    {d^k\mD{\lambda}}
    \leq1-\frac{kd(D-d)}{n}+O(n^{-2}).
\end{align}
\end{proposition}

\begin{proof}
    We assume $n\geq d$ and fix $\lambda\vdash_d n+k$.
    Expanding every $k$-step path into single-box steps gives
    \begin{align}
        \sum_{\alpha\vdash_d n}
        \frac{f^{\lambda/\alpha}\mD{\alpha}}
        {d^k\mD{\lambda}}
        &=
        \sum_{\substack{\alpha_1,\ldots,\alpha_k\\
        \alpha_i\in\alpha_{i+1}-\square\ (i=1,\ldots,k)\\
        \alpha_{k+1}=\lambda}}
        \prod_{i=1}^k
        \frac{\mD{\alpha_i}}
        {d\mD{\alpha_{i+1}}}.
        \label{eq:path-product}
    \end{align}
    By denoting $n+i=q_i d+r_i$ with $q_i, r_i\in \ZZ_{\geq 0}$ and $0\leq r_i<d$, Eq.~\eqref{eq:dimension-ratio} bounds each one-box sum as
    \begin{align}
        \sum_{\alpha_i\in\alpha_{i+1}-\square}
        \frac{\mD{\alpha_i}}{d\mD{\alpha_{i+1}}}
        &\leq{1\over d}
        \left[D-(D-d)R_{d,D,n+i-1}\right]\\
        &=1-\frac{D-d}{q_i}+O(q_i^{-2})\\
        &=1-\frac{d(D-d)}{n}+O(n^{-2})
    \end{align}
    for all $\alpha_{i+1}\vdash_d n+i$.
    Since $k,d,D$ are fixed, the $O(n^{-2})$ remainder is uniform over $i=1,\ldots,k$.
    Iterating the bound in Eq.~\eqref{eq:path-product} yields
    \begin{align}
        \sum_{\alpha\vdash_d n}
        \frac{f^{\lambda/\alpha}\mD{\alpha}}
        {d^k\mD{\lambda}}
        &\leq\left[1-\frac{d(D-d)}{n}+O(n^{-2})\right]^k\\
        &=1-\frac{kd(D-d)}{n}+O(n^{-2}).
    \end{align}
\end{proof}

\begin{proof}[Proof of optimality in Thm.~\ref{thm:multi-copy-conjugation-asymptotic}]
    Propositions~\ref{prop:multi-copy-optimality} and~\ref{prop:multi-copy-optimality-asymptotic} imply
    \begin{align}
        F_{d,D,n,k}^{(\mathrm{GEN})}
        &\leq\max_{\lambda\vdash_d n+k}
        \sum_{\alpha\vdash_d n}
        \frac{f^{\lambda/\alpha}\mD{\alpha}}
        {d^k\mD{\lambda}}\notag\\
        &\leq1-\frac{kd(D-d)}{n}+O(n^{-2}).
    \end{align}
    Together with Eq.~\eqref{eq:multi-copy-conjugation-asymptotic-lower-bound} in the proof sketch, this proves Eq.~\eqref{eq:multi-copy-conjugation-asymptotic}.
\end{proof}

\section*{Acknowledgment}

We acknowledge fruitful discussions with Adam Burchardt, Akihito Soeda, and thank Rashi Adhikari, Marco T\'{u}lio Quintino and Micha{\l} Studzi\'{n}ski for letting us know about their independent work~\cite{adhikari2026prep}.
We acknowledge the use of Codex 5.5 Extra High and 5.6-Sol Ultra for code generation of the SDP and numerical checks, which helped us to guess the primal and dual solutions in the proofs of Props.~\ref{prop:achievability} and \ref{prop:optimality} and prepare Fig.~\ref{fig:comparison}, and to assist the preparation of the manuscript.
All analytical arguments were independently derived and verified by the authors, and the authors take full responsibility for the content of the manuscript.
This work was supported by MEXT Quantum Leap Flagship Program (MEXT QLEAP) JPMXS0118069605 and JPMXS0120351339; Japan Science and Technology Agency (JST) as part of Adopting Sustainable Partnerships for Innovative Research Ecosystem (ASPIRE), Grant Number JPMJAP25A3; JST CREST, Grant Number JPMJCR25I5; JST NEXUS, Grant Number JPMJNX26C9; JSPS KAKENHI Grant No. 23K21643; and IBM Quantum.

\bibliographystyle{IEEEtran}
\bibliography{main}

\end{document}